\documentclass[10pt, twocolumn]{IEEEtran}

\ifCLASSINFOpdf
\usepackage[pdftex]{graphicx}
\else
\fi

\usepackage{amsfonts,amsmath,amsthm,epstopdf,cite,url,subfig}
\usepackage{psfrag,amssymb,amsmath,pifont,cite,graphics,graphicx,epsfig,amscd,helvet,multirow}
\usepackage[shortlabels]{enumitem}
\usepackage{mathtools,eqparbox}
\usepackage{empheq}

\usepackage{algorithm,algorithmic}
\usepackage{color}
\usepackage{epstopdf}
\usepackage{accents}
\usepackage{setspace}

\usepackage{cite,color}

\usepackage{caption}
\usepackage[disable]{todonotes}
\usepackage{microtype}

\graphicspath{{figs/}}

\usepackage{xcolor}
\usepackage[colorlinks, citecolor=red, linkcolor=darkgreen]{hyperref}

\definecolor{darkgreen}{RGB}{0, 100, 0}
\usepackage{cleveref}
\Crefname{figure}{Fig.}{Fig.}

\newtheorem{theorem}{Theorem}
\newtheorem{lemma}{Lemma}
\newtheorem{definition}{Definition}
\newtheorem{proposition}{Proposition}
\newtheorem{corollary}{Corollary}
\newtheorem{remark}{Remark}
\newtheorem{example}{Example}

\newif\ifhighlight
\highlightfalse

\newif\ifshowrebuttal
\showrebuttalfalse

\def\z{\mathbf{z}}

\def\x{\mathbf{x}}
\def\E{\mathbb{E}}

\def\v{\mathbf{v}}
\def\u{\mathbf{u}}

\def\f{\mathbf{f}}

\def\a{\mathbf{a}}
\def\e{\mathbf{e}}

\def\jj{\mathrm{j}}

\def\TT{{\mathcal T}}

\def\C{{\mathbb C}}

\def\xmax{{x_{(1)}}}
\def\xmin{{x_{(k)}}}
\def\xsupp{{\text{supp}(\x)}}

\def\etal{{\it et al. }}

\DeclareMathOperator*{\supp}{supp}
\DeclareMathOperator*{\dist}{dist}
\DeclareMathOperator*{\diag}{diag}

\newcommand{\ceiling}[1]{\lceil #1 \rceil}

\newcounter{todocounter}

\begin{document}

\title{Achieving Optimal Sample Complexity for a Broader Class of Signals in Sparse Phase Retrieval}

\author{
      Mengchu Xu, \textit{Member, IEEE}, Yuxuan Zhang, and Jian Wang, \textit{Member, IEEE}
      \thanks{
            The authors are with the School of Data Science, Fudan University, Shanghai 200433, China. The work was supported in part by the National Key R\&D Program of China under Grant 2024YFF0505601, and in part by the National Natural Science Foundation of China under Grant 62471147.}
      \thanks{Corresponding author: Jian Wang (e-mail: jian\_wang@fudan.edu.cn).}
}

\maketitle

\begin{abstract}
      Sparse phase retrieval aims to recover a $k$-sparse signal from $m$ phaseless measurements. While the theoretically optimal sample complexity for successful recovery is $\Omega(k \log n)$, existing algorithms can only achieve this bound for signals with specific structural assumptions, leading to a notable gap between theory and practice. To bridge this gap, we introduce an efficient initialization algorithm, termed generalized Exponential Spectral Pursuit (gESP). We prove that gESP can significantly expand the family of signals that are guaranteed to be recovered with the optimal sample complexity, thereby extending the scope of theoretical optimality to a much broader class of signals. Extensive simulations validate our theoretical findings and demonstrate that gESP consistently outperforms the state-of-the-art methods across diverse signal types.
\end{abstract}

\begin{IEEEkeywords}
      Sparse phase retrieval, sample complexity, sparsity, signal structure, information-theoretical bound.
\end{IEEEkeywords}

% !TEX root = 00Main.tex
\section{Introduction} \label{sec:I}
In many applications, only the intensity information is retained by the measurement system due to hardware limitations~\cite{Bunk2007Diff}. For example, in optical imaging~\cite{OpitcalLit,pinilla2023opt} and X-ray crystallography~\cite{XrayLit,Lit4}, common sensors such as charge-coupled devices, photosensitive films, and human eyes merely record the amplitude of the light wave but ignore the phase. In these settings, one is faced with the task of recovering the original signal from its magnitude-only measurements, which is often referred to as phase retrieval (PR)~\cite{WF,RWF,TrWF,PhaseCut,PhaseEqual}. Succinctly stated, the main task of PR is to reconstruct the $n$-dimensional signal $\x \in \mathbb{C}^{n}$ from a group of phaseless measurements $\mathbf{y} \in \mathbb{R}^{m}$ given by
\begin{equation}\label{PR}
      y_i = |\mathbf{a}_i^{*} \mathbf{x} |, \ i = 1,2,\cdots m,
\end{equation}
where $\mathbf{a}_i \in \mathbb{C}^n$ is the measurement vector. In this paper, we consider complex Gaussian measurements, i.e., $\mathbf{a}_i \sim \mathcal{CN}(n)$, and $\x$ is $k$-sparse, implying that there are at most $k$ non-zero entries in $\x$. This problem is known as the sparse phase retrieval problem. Researchers are interested in the minimum number of measurements required to reconstruct the signal~\cite{phaseliftoff,XuZQ,PRbyCSPhaseLift,AltMin}. It has been shown that with $m=4k-2$ phaseless measurements~\cite{Xu2013uniquespr}, the non-linear system~\eqref{PR} with a $k$-sparse prior produces a unique solution based on generic sensing measurements. Moreover, it has been proven that
\begin{equation} \label{eq:klogn}
      m=\Omega(k\log n)
\end{equation}
complex Gaussian measurements are sufficient to guarantee the stability of the solution~\cite{infobound1,infobound2}. This sample complexity is often called the information-theoretical bound for sparse phase retrieval.

For polynomial-time algorithms performing exact recovery of $\mathbf{x}$, it appears to require significantly more phaseless measurements. Back in 2013, Li and Voroninski~\cite{SparsePR-SDP} utilized a matrix lifting technique and convex relaxation methods to solve the sparse phase retrieval problem. They proved that it requires \begin{equation}\label{eq:k2logn-classical}
      m=\Omega(k^2\log n)
\end{equation}
measurements to produce an exact reconstruction. Moreover, they further proved that this requirement cannot be reduced to be smaller than $
      m=\Omega({k^2}/{\log^2 n})
$ for their supposed model. While their sample complexity revealed a clear discrepancy compared to \eqref{eq:klogn}, the authors were unsure whether this difference was inherent to the problem itself or simply a consequence of limitations in their modeling approach. Later, this discrepancy was identified as the statistical-to-computational gap~\cite{infobound2, sunju}, particularly in the context of generic Gaussian measurements. Nevertheless, it has been shown that this barrier can be circumvented by imposing specific structures on the measurement process. Prior works, such as~\cite{bahmani2015efficient, jaganathan2013sparse}, demonstrated that the optimal sample complexity is achievable in polynomial time when additional constraints are placed on the sensing vectors. However, under the standard unstructured Gaussian setting, verifying the existence of this gap remains a significant challenge.

This gap was further investigated in the thresholded Wirtinger Flow (TWF) algorithm~\cite{ThWF}, which is a typical greedy algorithm containing two stages: the initialization stage and the refinement stage. The authors showed that both stages require $m=\Omega(k^2\log n)$ samples for exact final recovery. This two-stage approach was followed in subsequent algorithms, e.g., SPARse Truncated Amplitude flow (SPARTA)~\cite{SPARTA}, Compressive Phase Retrieval with Alternating Minimization (CoPRAM)~\cite{CoPRAM}, and Subspace Phase Retrieval (SPR)~\cite{SPR}. It is shown via some geometric analysis~\cite{sunju, SPR,quadraticeq} and sophisticated optimization approaches~\cite{StructuredSignalRecovery} that the sample complexity of the refinement stage can be improved to the information-theoretical bound~\eqref{eq:klogn}, provided that a $\delta$-neighborhood estimate of the target signal is available.
However, to achieve this, the initialization stage requires
\begin{equation}
      m=\Omega(k^2\log n)
\end{equation}
measurements, which is still far from the information-theoretical bound~\eqref{eq:klogn}. It has become a consensus that the statistical-to-computational gap results from the initialization stage~\cite{sunju, CJF}. More precisely, it is the requirement of finding an estimate falling into the $\delta$-neighborhood of $\x$ that leads to the statistical-to-computational gap.

Despite the persistence of this gap, numerous efforts are underway to achieve improved results. For the sake of notational simplicity, we define the structural function $s(p)$ as
\begin{equation}
      s(p) = \frac{\|\x\|^2}{\sum_{j=1}^{p} |x_{(j)}|^2},~~ p=1,2,\cdots,n,
\end{equation}
where $x_{(1)},x_{(2)}, \cdots x_{(k)}, \cdots x_{(n)}$ are the rearrangement of the entries of $\x$ in descending order of their magnitudes. A formal definition and the corresponding interpretation are presented in Section~\ref{sec:3.1}.
Wu and Rebeschini~\cite{HWF} presented an algorithm called Hadamard Wirtinger Flow (HWF).
They showed that the gap vanished for signals with specific structures. Specifically, by assuming that the minimum magnitude of non-zeros $x_{(k)}$ is on the order of $\frac{\|\x\|}{\sqrt{k}}$, i.e.,
\begin{equation}\label{eq:xmincondition}
      \left\lvert x_{(k)} \right\rvert  =  \Omega \left(\frac{\|\x\|}{\sqrt{k}}\right),
\end{equation}
they showed that the sample complexity can be reduced to
\begin{equation}\label{HWF-comp}
      m = \Omega\left(\max \left\{ ks(1) \log n,  \sqrt{ks(1)} \log^3 n \right\} \right).
\end{equation}
This means that when the maximum entry $x_{(1)}$ of $\x$ is on the order of $\|\x\|$ (ignoring the $\log^3 n$ term), \eqref{HWF-comp} can be further reduced to~\eqref{eq:klogn}, and thus reach the theoretical lower bound.

The result of~\eqref{HWF-comp} was further improved by Cai \etal \cite{CJF} and Xu \etal \cite{ESP}. In particular, Cai \etal \cite{CJF} proposed the Truncated Power Method (TPM) to remove the requirements on the minimum magnitude of non-zeros (i.e., Eq.~\eqref{eq:xmincondition}). Xu~\etal~\cite{ESP} introduced an exponential spectrum and proposed the Exponential Spectral Pursuit (ESP) algorithm to eliminate the $\log^3 n$ term in~\eqref{HWF-comp}. The sample complexity was reduced to
\begin{equation}\label{ESP-comp}
      m = \Omega\left(  ks(1) \log n  \right),
\end{equation}
which is sufficient for ESP to produce a good estimate falling into the $\delta$-neighborhood of $\x$ without any assumption on $x_{(k)}$. However, this result only achieves the optimal bound~\eqref{eq:klogn} for a restrictive class of signals, i.e., those with a single dominant entry where $s(1)=\Theta(1)$. This limitation motivates the development of algorithms that can achieve optimal sample complexity for a broader class of signals.

In this paper, we develop the generalized Exponential Spectral Pursuit (gESP) algorithm to address this challenge. Specifically, gESP selects {\it an index set} rather than a single index in its first step (i.e., \textbf{Step 1} in \Cref{alg:gESP}). This key modification allows gESP to achieve the optimal sample complexity for a much broader class of signals. Specifically, it relaxes the condition for optimality from the restrictive $s(1)=\Theta(1)$ to the more general $s(\ceiling{\sqrt{k}})=\Theta(1)$. We show that when the signal structure $s(p)$ is known, gESP produces a $\delta$-neighborhood estimate with
\begin{equation}\label{gESP-knownp}
      m = \Omega \left(\min_{p \in [k]} \max \left\{p^2 s^2(p), ks(p) \right\} \log n \right)
\end{equation}
samples. As will be detailed in Section~\ref{sec:dis}, this result is uniformly better than \eqref{ESP-comp}.

An alternative approach is also provided if we have no access to $s(p)$; see Corollary~\ref{coro:3}. In this case, the sample complexity is
\begin{equation}\label{gESP-comp}
      m = \Omega \left( \min_{p \in [\ceiling{\sqrt k}]} \max \left\{p^2 s^2(p) ,\sqrt{k}s^2(p),  ks(p)  \right\} \log n \right).
\end{equation}
It can be shown that in the case where $s(1)=\Theta(\sqrt{k})$ and $s(\sqrt{k})=\Theta(1)$, this result will also be better than \eqref{ESP-comp}. A detailed discussion can be found in Section~\ref{sec:dis}.

The remainder of this paper is organized as follows. In Section~\ref{sec:2}, we review the development of the initialization stage and introduce our algorithm. Detailed interpretations are provided step by step. Section~\ref{sec:3} contains theoretical results for the sample complexity of gESP in different cases and makes simple comparisons between these results. In Section~\ref{sec:dis}, we discuss various aspects to illustrate the superiority of our result. In Section~\ref{sec:proof}, the proofs of the main theorems are given based on the analysis for each step of gESP. Numerical simulations and analysis are conducted in Section~\ref{sec:simu}. Finally, we conclude our paper in Section~\ref{sec:conc}.

% !TEX root = 00Main.tex
\section{Algorithms}\label{sec:2}
We introduce some useful notations. Throughout the paper, we use lowercase and uppercase boldface letters to represent vectors and matrices, while employing normal font for real numbers. Let $[n]$ denote the set $\{1,2,3,\cdots n\}$.
For the sake of readability, we denote $x_{(1)},x_{(2)}, \cdots x_{(k)}, \cdots x_{(n)}$ as the rearrangement of the entries of $\x$ in descending order of their magnitudes. For any set $S \subset [n]$, $\C^{S}$ is defined as the subspace of $\C^{n}$ spanned by vectors supported on $S$, i.e., $\{\x|\x \in \C^{n}, \xsupp = S\}$. Unless otherwise specified, we define $\mathbf{a}_{S}$ as the vector that keeps the entries of $\mathbf{a}$ indexed by $S$ while setting others to zero. For the matrix $\mathbf{A}$, we define $\mathbf{A}_{S}$ as the matrix which keeps columns and rows indexed by $S$ while setting others to zero. Given a vector $\x$,  the conjugate transpose, $\ell_2$ norm, and $\ell_0$ norm of $\x$ are denoted as $\x^{*}$, $\|\x\|$ and $\|\x\|_0$ respectively. An $n$-dimensional standard complex Gaussian random vector, denoted as $\a \sim \mathcal{CN}(n)$, is defined as $\a = \a_1 + \jj \a_2$, with $\a_1, \a_2 \sim \mathcal{N}(\mathbf{0}, \frac{1}{2} \mathbf{I})$, and $\text{j}$ is the imaginary unit. We use $a\sim b$ to represent that $a$ differs from $b$ by a constant factor. Throughout this paper, $c$ and $C$ denote universal positive constants whose values may change from line to line. We say an event occurs with high probability if it holds with probability at least $1 - n^{-c}$ for some constant $c > 0$.
Following the convention in the literature, we use standard asymptotic notations to characterize the complexity as follows:

\begin{itemize}
      \item \textbf{Big-Omega Notation ($\Omega$):} We say $f(n) = \Omega(g(n))$ if there exist positive constants $c$ and $n_0$ such that for all $n \ge n_0$, we have $0 \le c \cdot g(n) \le f(n)$. This denotes an \textit{asymptotic lower bound}.

      \item \textbf{Big-Theta Notation ($\Theta$):} We say $f(n) = \Theta(g(n))$ if there exist positive constants $c_1, c_2,$ and $n_0$ such that for all $n \ge n_0$, we have $0 \le c_1 \cdot g(n) \le f(n) \le c_2 \cdot g(n)$.
\end{itemize}

Before introducing the proposed algorithm, we first review the prior work. In this section, we temporarily assume that \begin{equation}\label{eq:xmin-assumption}
      \left\lvert \x_{(k)} \right\rvert  = \Omega\left(\frac{\|\x\|}{\sqrt{k}}\right).
\end{equation}
While this assumption is made for the convenience of our brief introduction, it will be shown in our main theorems that this assumption is actually unnecessary; Similar arguments can be found in~\cite{CJF,ESP}.

\subsection{Prior work}
As introduced in Section~\ref{sec:I}, the initialization stage aims to produce a desired estimate that falls into the $\delta$-neighborhood of the target signal $\x$.
\begin{enumerate}
      \item Earlier initialization approaches typically involve two steps:

            \vspace{2mm}\textbf{Step 1:} {\it Construct a spectrum $\mathbf{Z}$ and take the diagonal entries to be the set $\{Z_j\}_{j=1}^n$. Sorting  $\{Z_j\}_{j=1}^n$ and taking the first $k$ indices yields the final support estimate $S$.}\vspace{2mm}

            Generally, the expectation of $\mathbf{Z}$ takes the following form:
            \begin{equation}\label{eq:Zexp}
                  \E[\mathbf{Z}] = \alpha \|\x\|^2\mathbf{I}+ \beta  \x\x^*.
            \end{equation}
            Thus, the expectation of $Z_j$ takes the following form
            \begin{equation}\label{eq:EofZ}
                  \E[Z_j] = \alpha \|\x\|^2 + \beta |x_j|^2,
            \end{equation}
            where $\alpha,\beta$ are both numerical constants, determined by the specific definition of $Z_j$. For instance, \cite{SPARTA,SWF} set $Z_j = \frac{1}{m} \sum_{i = 1}^{m}y_i^2 | \a_{ij}|^2$ with the corresponding $\alpha=1, \beta=2$, whereas~\cite{AltMin} set $Z_j = \frac{1}{m} \sum_{i = 1}^{m}y_i |\a_{ij}|$ with the corresponding $\alpha=\frac{2}{\pi\|\x\|}, \beta\approx\frac{1}{6\|\x\|}$.  It can be seen that there exists an apparent gap between $\{\E[Z_j]\}_{j\in \supp(\x)}$ and $\{\E[Z_j]\}_{j\in \supp(\x)^c}$. Specifically,
            \begin{align}
                  \nonumber  \min_{j\in \supp(\x)} \E[Z_j] - \max_{j\in \supp(\x)^c} \E[Z_j] = & ~ \beta |\xmin|^2                          \\
                  \sim                                                                         & ~ \frac{\|\x\|^2}{k}.\label{eq:suppornot1}
            \end{align}
            Thus, employing some concentration inequalities on $Z_j$ and $\E[Z_j]$ finally leads to
            \begin{equation}\label{eq:suppornot2}
                  \min_{j\in \supp(\x)} Z_j  > \max_{j\in \supp(\x)^c} Z_j,
            \end{equation}
            which means one can obtain the estimate of support set by sorting the elements of $\{Z_j\}_{j=1}^n$ and selecting the indices corresponding to the $k$ largest values. This guarantees that the support estimate contains all the true candidates. Theoretical analysis shows that deriving an appropriate concentration from the gap $\frac{\|\x\|^2}{k}$ (see \eqref{eq:suppornot1}) requires the sample complexity being\footnote{When $Z_j$ is statistically heavy-tailed (e.g., involving the fourth power of Gaussian), some extra term containing $\log^3 n$ (e.g., see \eqref{HWF-comp}) will appear in the complexity. This is because one needs to use a truncation argument to analyze the tail bound. For simplicity, we will ignore such term throughout this subsection.}
            \begin{equation}\label{eq:twostep1}
                  m = \Omega(k^2 \log n),
            \end{equation}
            which results from the re-scaled gap $\frac{|\xmin|^2}{\|\x\|^2}\leq \frac{1}{k}$.

            \vspace{2mm}\textbf{Step 2:} {\it Set $\z$ as the leading eigenvector of $\mathbf{Z}_S$ with $\|\z\|^2 = \lambda^2 \doteq \frac{1}{m} \sum_{i=1}^{m} y_i^2$.} \vspace{2mm}

            The leading eigenvector of $\E[\mathbf{Z}] $ is $\x$ and $\E[\|\z\|^2] = \|\x\|^2$. Therefore, we shall obtain a good estimate of $\x$. Theoretical analysis shows that this step requires
            \begin{equation}\label{eq:twostep2}
                  m = \Omega(k \log n)
            \end{equation}
            samples. Combining \eqref{eq:twostep1} and \eqref{eq:twostep2}, the final sample complexity is \begin{equation}
                  m = \Omega(k^2 \log n).
            \end{equation}

      \item Recently, a novel initialization paradigm was developed in~\cite{HWF, CJF, ESP}, which contains four steps. Instead of directly estimating the support in one step, an index $j_{\text{max}}$ is first selected corresponding to the largest value among $Z_j$'s.

            \vspace{2mm}\textbf{Step 1:} {\it Construct a spectrum $\mathbf{Z}$ and take the diagonal entries to be the set $\{Z_j\}_{j=1}^n$. Select the index $j_{\text{max}}$ corresponding to the largest value among $Z_j$'s.}\vspace{2mm}

            This step is proved to guarantee
            \(
            x_{j_{\text{max}}}\sim \xmax
            \)
            when
            \begin{equation}\label{eq:maxstep1}
                  m = \Omega\left(\frac{\|\x\|^4}{|\xmax|^4} \log n\right).
            \end{equation}

            \vspace{2mm}\textbf{Step 2:} {\it Construct a unit vector $\mathbf{e}_{j_{\text{\rm max}}}$, whose values take $1$ at $j_{\text{max}}$ and $0$ at others.

            \vspace{2mm}\textbf{Step 3:} {\it Generate an index set $S^1$ corresponding to the largest $k$  entries of the vector $   \mathbf{Z}\mathbf{e}_{j_{\text{{\rm max}}}} $ in modulus.}\footnote{In previous literature, this step and the subsequent one (i.e., {\bf Steps 2} \& {\bf 3}) were typically merged into a single step: {\it generating an index set $S^1$ corresponding to the largest $k$ entries of the $j_\text{max}$ column of $\mathbf{Z}$.} Here, we have separated them to align with the structure of our proposed algorithm (Algorithm~\ref{alg:gESP}).}}\vspace{2mm}

            The expectation of $\mathbf{Z}$ still satisfies \eqref{eq:Zexp}. Thus, we have
            \begin{align}
                  \mathbf{Z}\mathbf{e}_{j_{\text{max}}} & \approx   \E[\mathbf{Z}] \mathbf{e}_{{j_{\text{max}}}}  \nonumber  \\
                                                        & =  \alpha \|\x\|^2 \mathbf{e}_{j_{\text{max}}} + \beta \xmax  \x .
            \end{align}

            Establishing some concentration inequalities on this random variable yields a result similar to~\eqref{eq:suppornot2}:
            \begin{align}
                  \nonumber \min_{j\in \supp(\x)}   |\mathbf{Z}\mathbf{e}_{j_{\text{max}}} |_j  - \max_{j\in \supp(\x)^c}   |\mathbf{Z}\mathbf{e}_{j_{\text{max}}} |_j \approx \  & \beta |\xmax \xmin|                                 \\
                  \sim \                                                                                                                                                          & \frac{|\xmax|\|\x\|}{\sqrt{k}}.\label{eq:Zejmaxgap}
            \end{align}

            To identify the true support set based on \eqref{eq:Zejmaxgap}, it requires
            \begin{equation}
                  m = \Omega\left(k\frac{\|\x\|^2}{|\xmax|^2} \log n\right).
            \end{equation}

            \vspace{2mm}\textbf{Step 4:} {\it Set $\z$ as the leading eigenvector of $\mathbf{Z}_{S^1}$ with $\|\z\|^2 = \lambda^2 \doteq \frac{1}{m} \sum_{i=1}^{m} y_i^2$.} \vspace{2mm}

            This step is the same as {\bf Step 2} of the previous method.

            Finally, the overall sample complexity of this method is
            \begin{equation}
                  m = \Omega\left(k\frac{\|\x\|^2}{|\xmax|^2} \log n\right).
            \end{equation}
            Since $\frac{\|\x\|^2}{|\xmax|^2}\leq k$, the enhancement of this method lies in the enlarged gap (see \eqref{eq:Zejmaxgap}) $\frac{|\xmax|\|\x\|}{\sqrt{k}}$, which is potentially larger than  $\frac{\|\x\|^2}{k}$.

\end{enumerate}

\subsection{Generalized Exponential Spectral Pursuit}

As has been clarified, the enlarged gap $\frac{|\xmax|\|\x\|}{\sqrt{k}}$ leads to the improvement in complexity. Intuitively, the larger the gap is, the smaller the sample complexity will be. Note that the term $|\xmax|$ in this gap arises from the fact that we simply select {\it one} maximal index in {\bf Step 1}. If more significant indices are chosen in {\bf Step 1}, the final gap will indeed be enlarged, thereby decreasing the complexity in {\bf Step 3}\footnote{In fact, the complexity in {\bf Step 1} will increase while that in {\bf Step 3} will decrease as more significant indices are chosen. This means that there exists a trade-off in the number of the indices that should be chosen in {\bf Step 1}.}. Then the remaining issue is to design {\bf Step 2}, where the set $S^0$ from {\bf Step 1} may contain more than one support candidate. Following the previous method, a straightforward idea is to construct the vector $\mathbf{e}^0$ such that it takes $1$ on $S^0$ and $0$ for the rest. A more sophisticated design, however, is to compute the leading eigenvector of $\mathbf{Z}_{S^0}$ as $\mathbf{e}^0$. This also matches the previous method since $\mathbf{e}_{j_{\text{max}}}$ is exactly the leading eigenvector of $\mathbf{Z}_{\{j_{\text{max}}\}}$.

So far, we have completed the introduction of the proposed modification. The specific algorithm is summarized in Algorithm~\ref{alg:gESP}. Now we will further clarify the effectiveness intuitively. The more detailed sample complexity and numerical conclusions
are formally presented in Section~\ref{sec:3}.

Following~\cite{ESP}, we utilize the spectrum:
\begin{equation}
      \mathbf{Z}=\frac{1}{m} \sum_{i=1}^{m} \left(\frac{1}{2} - \exp\left(- \frac{y_i^2}{\lambda^2} \right) \right) \mathbf{a}_i \mathbf{a}_{i}^{*},
\end{equation}
where $\lambda^2 = \frac{1}{m}\sum_{i=1}^{m}y_i^2$. It can be shown that $\mathbf{Z}$ satisfies
\begin{equation}
      \E[\mathbf{Z}] \approx \alpha \|\x\|^2\mathbf{I}+ \beta  \x\x^*.
\end{equation}
and the diagonal entries $Z_j$'s satisfies
\begin{equation}
      \E[Z_j] \approx \alpha \|\x\|^2 + \beta |x_j|^2,
\end{equation}

{\bf Step 1} ensures that \begin{equation}
      \|\x_{S^0}\|^2\sim \sum_{i=1}^{p} |x_{(i)}|^2.
\end{equation}
Since the leading eigenvector of $\E[\mathbf{Z}_{S^0}]$ is $\x_{S^0}/ \|\x_{S^0}\|$, we can approximate  $\e_0$ as $\x_{S^0}/ \|\x_{S^0}\|$. Then,
{\bf Step 2} ensures that \begin{equation}
      |\x^*\e_0| = |\x_{S^0}^*\e_0| \approx \|\x_{S^0}\| \sim \sqrt{\sum_{i=1}^{p} |x_{(i)}|^2}.
\end{equation}
Note that
\begin{align}
      \mathbf{Z} \e^0 \approx \E[\mathbf{Z}] \e^0 = \alpha \|\x\|^2\e^0 + \beta  \x^*\e^0\cdot \x.
\end{align}
Similar to \eqref{eq:Zejmaxgap}, we have
\begin{align}
      ~        & \min_{j\in \supp(\x)}    |\mathbf{Z}\mathbf{e}^0 |_j  - \max_{j\in \supp(\x)^c}  |\mathbf{Z}\mathbf{e}^0 |_j \nonumber \\
      \approx~ & \beta |\x^*\e^0| \cdot \xmin \sim \frac{ \sqrt{\sum_{i=1}^{p} |x_{(i)}|^2}\|\x\|}{\sqrt{k}}.\label{eq:Zejmaxgap2}
\end{align}
This means that in {\bf Step 3}, we can distinguish the true support entries. Then
      {\bf Step 4} can ensure that $\z$ is a good initialization of the target signal $\x$.

\begin{algorithm}[t]
      \caption{~Generalized Exponential Spectral Pursuit} \label{alg:gESP}
      \begin{algorithmic}
            \STATE  \parbox{0.2\linewidth}{\bf Input:}    \parbox{0.79\linewidth}{$\mathbf{A}$, $\mathbf{y}$, the sparsity $k$, the parameter $p$.}
            \hspace{5mm}

            \STATE \parbox{0.2\linewidth}{\bf Step 1:\\} \parbox{0.79\linewidth}{Generate an index set $S^0$ corresponding to the largest $p$ diagonal elements of $\mathbf{Z}$.}
            ~\\

            \STATE \parbox{0.2\linewidth}{\bf Step 2:} \parbox{0.79\linewidth}{Set $\mathbf{e}^0$ as the unit leading eigenvector of $\mathbf{Z}_{S^0}$.}
            ~\\

            \STATE \parbox{0.2\linewidth}{\bf Step 3:\\~\\} \parbox{0.79\linewidth}{Generate an index set $S^{1}$ corresponding to the largest $k$  entries of the vector $\f = \mathbf{Z} \e^0$ in modulus.}
            ~\\
            ~\\

            \STATE \parbox{0.2\linewidth}{\bf Step 4:\\}  \parbox{0.79\linewidth}{Set $\z$ as the leading eigenvector of $\mathbf{Z}_{S^{1}}$ with $\|\z\|^2 = \lambda^2 \doteq \frac{1}{m} \sum_{i=1}^{m} y_i^2$. }
            ~\\

            \STATE \parbox{0.2\linewidth}{\bf Output:}    \parbox{0.79\linewidth}{$\z$.}
      \end{algorithmic}
\end{algorithm}

% !TEX root = 00Main.tex
\section{Theoretical Analysis}\label{sec:3}
\subsection{Preliminary}\label{sec:3.1}

To begin with, we define the distance between two complex vectors.
\begin{definition}
      For any $n$-dimensional complex vectors $\u, \v \in \C^{n}$, the distance between $\u$ and $\v$ is defined as
      \begin{equation}
            \dist(\u, \v) = \min_{\phi \in [0,2 \pi)} \| \u - e^{\jj \phi} \v \|,
      \end{equation}
      where $\jj$ is the imaginary unit.
\end{definition}
In the initialization stage, typical algorithms, such as HWF, SPARTA and TPM, usually aim to produce a good estimate $\z$ falling into the $\delta$-neighborhood of the target signal, i.e.,
\begin{equation}\label{initial-target}
      \dist(\z, \x) \leq \delta \|\x\|,
\end{equation}
which is the prerequisite of the refinement stage~\cite{ThWF, SPARTA,SWF}. In the analysis of gESP, we follow the same line and focus on obtaining an estimate satisfying~\eqref{initial-target}.

Furthermore, we give the definition of the structure function $s(p)$ of the target signal $\x$.
\begin{definition}\label{def:s}
      Let $x_{(1)},x_{(2)}, \cdots x_{(k)}, \cdots x_{(n)}$ be the rearrangement of the entries of $\x$ in descending order of their magnitudes. The structure function $s(p)$ is given by
      \begin{equation}\label{def-s}
            s(p) = \frac{\|\x\|^2}{\sum_{j=1}^{p} |x_{(j)}|^2},~~ p=1,2,\cdots,n.
      \end{equation}
\end{definition}
The structure function, $s(p)$, quantifies the structural information of a signal $\mathbf{x}$ by measuring the energy distribution across its components. Specifically, $s(p)$ represents the inverse of the proportion of the total signal energy captured by the $p$ largest magnitude components. Therefore, the behavior of $s(p)$ as $p$ varies directly reflects how energy is concentrated or dispersed within the signal, thus characterizing its structure.

The function $s(p)$ has two important properties that are instrumental for the subsequent analysis.
\begin{enumerate}
      \item $s(p)$ is monotonically {\it decreasing} with respect to $p$. Furthermore,
            \begin{equation}\label{range-s}
                  1 \leq s(p)  \leq \frac{k}{p}.
            \end{equation}
            \vspace{2mm}

      \item   Note that
            \begin{equation}
                  p s(p) = \frac{\|\x\|^2}{\frac{1}{p} \sum_{i=1}^{p} |x_{(i)}|^2},
            \end{equation}
            which implies $p s(p)$ is monotonically {\it increasing} with respect to $p$. Furthermore,
            \begin{equation}\label{range-ps}
                  p \leq p s(p) \leq k.
            \end{equation}

\end{enumerate}

Finally, to facilitate rigorous theoretical analysis and decouple statistical dependencies among different iterations, we employ a standard sample splitting strategy \cite{jain2013low,WF} in our proof. Specifically, the measurement set $\mathcal{D} = \{(a_i, y_i)\}_{i=1}^m$ is randomly partitioned into $4$ disjoint subsets $\mathcal{D}_1, \dots, \mathcal{D}_4$, where each subset is used for a single step of our algorithm. Here we emphasize that sample splitting is purely a theoretical artifact used to simplify the probabilistic analysis. In our practical implementation (\Cref{alg:gESP}) and numerical experiments, we utilize the full dataset $\mathcal{D}$ at every step to maximize the sample efficiency. Since the number of steps is small and fixed (i.e., $4$), splitting the data only affects the sample complexity by a constant factor, leaving the order of the convergence rate unchanged.

\subsection{Theoretical results for sample complexity of gESP}\label{sec:3.2}
In this subsection, we give the theoretical results for Algorithm~\ref{alg:gESP}. The corresponding proofs are deferred to Section~\ref{sec:proof}.

First, we present the theoretical result for Algorithm~\ref{alg:gESP}, and provide a specific value of $p$ when $s(p)$ is known.
\begin{theorem}\label{thm:1}
      For any constant $0 < \delta < 1$, the output $\z$ of Algorithm~\ref{alg:gESP} satisfies $\dist(\z,\x) \leq \delta \|\x\|$ with probability at least $1 - n^{-c}$.
      The corresponding required sample complexity depends on an input parameter $p \in [k]$; for a general $p$, it is given by
      \begin{equation}\label{eq:generalresult-gESP}
            m = \Omega\left(\max\left\{p^2s^2(p), ks(p)\right\}\log n\right).
      \end{equation}
      To achieve the most efficient sampling, this complexity can be minimized by selecting the optimal parameter $p_{\text{opt}}$ that solves
      \begin{equation}\label{eq:pstar}
            p_{\text{opt}} = {\arg \min}_{p \in [k]} \max \left\{p^2 s^2(p), ks(p) \right\}.
      \end{equation}
      By employing this optimal parameter, the minimal sample complexity for the algorithm becomes
      \begin{equation}\label{eq:gESP_p_result}
            m = \Omega \left( \min_{p \in [k]} \max \left\{p^2 s^2(p), ks(p)  \right\} \log n \right).
      \end{equation}
\end{theorem}

\begin{remark}\label{remark1}
      The state-of-the-art result \cite{ESP} is the special case $p=1$ of our result, which is
      \begin{align}
            m & = \Omega \left(\max \left\{ s^2(1) \log n, ks(1) \log n \right\} \right)\nonumber \\
              & =\Omega\left(ks(1) \log n\right).
      \end{align}
      Therefore, our result~\eqref{eq:gESP_p_result} exhibits a significant improvement.
\end{remark}

\begin{remark}\label{remark:worst-case}
      The statistical-to-computational gap persists for signals with maximally delocalized energy. Specifically, the worst-case sample complexity of $\Omega(k^2\log n)$ is required if and only if the signal is almost ``flat'', i.e., $s(p) = \Theta\left(\frac{k}{p}\right)$ for all $p \in [1, k]$.
\end{remark}

\begin{remark}\label{remark2}
      To make~\eqref{eq:gESP_p_result} achieve the theoretical sampling minimum $\Omega(k\log n)$, the following inequalities must hold for some $p\in[k]$:
      \begin{empheq}[left=\empheqlbrace]{align}
            s(p)\leq \Theta(1),\label{eq:remark2-1}\\
            p^2s^2(p)\leq \Theta(k).\label{eq:remark2-2}
      \end{empheq}
      Since $s(p)\geq 1$ (see~\eqref{range-s}),~\eqref{eq:remark2-1} implies $s(p)= \Theta(1)$ for some $p \in [k]$. Then,~\eqref{eq:remark2-2} requires $p\leq \ceiling{c\sqrt{k}}$ for some constant $c$.
      Moreover, if $s(p)= \Theta(1)$ holds for some fixed $p\leq \ceiling{c\sqrt{k}}$, then $s(\ceiling{\sqrt{k}}) = \Theta(1)$ since $s(p)$ is monotonically decreasing and larger than $1$. Therefore, to achieve the information-theoretical bound~\eqref{eq:klogn}, gESP requires $s(\ceiling{\sqrt{k}}) = \Theta(1)$.
\end{remark}
This result provides the general sample complexity for gESP.  Interestingly, this result depends solely on the cumulative energy profile $s(p)$, imposing no lower bound requirement on the minimum non-zero magnitude $|x_{(k)}|$. This is due to that the assumption \eqref{eq:xmin-assumption} is required for exact support recovery, whereas our analysis only targets an $\ell_2$-distance bound ($\dist(\z,\x) \le \delta \|\x\|$), which, therefore, can be met under a relaxed requirement. A more specific explanation is provided below~\Cref{prop:3}.

In the case where $s(p)$ is unknown, determining the optimal $p$ is challenging. One compromise is to execute Algorithm~\ref{alg:gESP} $k$ times with $p$ taking each value in $[k]$. This yields a set containing $k$ estimates of the target signal $\x$. Naturally, at least one estimate falls into the $\delta$-neighborhood of $\x$ with the same sample complexity requirement as~\eqref{eq:gESP_p_result}, which leads to  the following corollary.
\begin{corollary}\label{coro:1}
      Run Algorithm~\ref{alg:gESP} $k$ times with $p$ taking each value in $[k]$ in turn, and collect all the estimates into a set $\mathcal X$. Then, for any $0 < \delta < 1$, with probability at least $1-n^{-c}$, there exists at least one estimate $\z\in \mathcal X$ that satisfies $\dist(\z,\x) \leq \delta \|\x\|$ when
      \begin{equation}\label{thm-eq2}
            m = \Omega \left( \min_{p \in [k]} \max \left\{p^2 s^2(p) \log n, ks(p) \log n \right\} \right).
      \end{equation}
\end{corollary}

\begin{remark}\label{remark:heuristic}
      While \Cref{coro:1} guarantees the optimal sample complexity by iterating $p$ over the entire range of $[1, k]$, this approach entails running the algorithm $k$ times. To mitigate the computational cost in practice, one can adopt a logarithmic search strategy, namely, selecting candidates from a subset $\mathcal{P} = \{1, 2, 4, \dots, 2^{\lfloor \log_2 k \rfloor}\}$. Since the complexity bound varies smoothly with $p$, this heuristic typically identifies a parameter that yields a sample complexity within a constant factor of the global optimum, while significantly reducing the computational overhead from $O(k)$ to $O(\log k)$ rounds.
\end{remark}

Due to computational constraints, multiple executions of the algorithm with varying $p$ are impractical. Thus, we are motivated to explore the sample complexity for specific, analytically tractable choices of $p$. A natural starting point for such an analysis is the intuitive choice of setting $p=k$. While this can be viewed as a special case of Theorem~\ref{thm:1}, we demonstrate that a more refined analysis allows us to derive a sharper result, as presented in the following corollary.
\begin{corollary}\label{coro:2}
      Take $p=k$ as the input of Algorithm~\ref{alg:gESP}. Then, for any $0 < \delta < 1$, with probability at least $1-n^{-c}$, the output $\z$ of Algorithm~\ref{alg:gESP} satisfies $\dist(\z,\x) \leq \delta \|\x\|$ provided that the number of samples $m$ satisfies
      \begin{equation}\label{eq:gESP_k_result}
            m = \Omega \left( k s^2(\lceil{\sqrt{k}}\rceil) \log n \right).
      \end{equation}
\end{corollary}

One may note that this result exactly coincides with the value in \eqref{eq:generalresult-gESP} with $p=\ceiling{\sqrt{k}}$. Interestingly, by setting $p=\ceiling{\sqrt{k}}$ as the input, we can obtain a better result, as demonstrated in the following theorem.

\begin{corollary}\label{coro:3}
      Take $p=\ceiling{\sqrt k}$ as the input of Algorithm~\ref{alg:gESP}. Then, for any $0 < \delta < 1$, with probability at least $1 - n^{-c}$, the output $\z$ of Algorithm~\ref{alg:gESP} satisfies $\dist(\z,\x) \leq \delta \|\x\|$ provided that the number $m$ of samples satisfies
      \begin{equation}\label{eq:gESP_sqrtk_result}
            m = \Omega \left( \min_{p \in [\ceiling{\sqrt k}]} \max \left\{p^2 s^2(p) ,\sqrt{k}s^2(p),  ks(p)  \right\} \log n \right).
      \end{equation}
\end{corollary}

On one hand, the result~\eqref{eq:gESP_sqrtk_result} cannot be worse than~\eqref{eq:gESP_k_result} since
\begin{align}
           & ~\Omega \left(\min_{p \in [\ceiling{\sqrt k}]} \max \left\{p^2 s^2(p) ,\sqrt{k}s^2(p),  ks(p)  \right\} \log n\right) \nonumber \\
      \leq & ~ \Omega \left(\max \left\{p^2 s^2(p) ,\sqrt{k}s^2(p),  ks(p)  \right\} \Big |_{p=\ceiling{\sqrt k}} \log n\right) \nonumber    \\
      =    & ~\Omega \left( k s^2(\lceil{\sqrt{k}}\rceil) \log n \right).
\end{align}
Moreover, there exist some signals that make the inequality strictly hold. Therefore, \eqref{eq:gESP_sqrtk_result} is indeed tighter than~\eqref{eq:gESP_k_result}. See Section~\ref{sec:DissB} for an in-depth discussion.

On the other hand, compared to the state of the art, the results~\eqref{eq:gESP_k_result} and~\eqref{eq:gESP_sqrtk_result} are superior for certain signals. In other words, even though $s(p)$ is unknown and executing Algorithm~\ref{alg:gESP} multiple times is infeasible, our algorithm taking $p=k$ or $\ceiling{\sqrt k}$ can perform better compared with existing algorithms. \Cref{sec:disA} provides an in-depth discussion.

% !TEX root = 00Main.tex
\section{Discussion}\label{sec:dis}
\subsection{Broader Conditions for Optimal Sample Complexity}\label{sec:disA}

As discussed in Remark~\ref{remark1} and Remark~\ref{remark2}, gESP offers a significant improvement over the state-of-the-art result, provided that the structural function $s(p)$ is known. This subsection further explores this superiority, demonstrating that gESP achieves the optimal sample complexity for a broader class of signals.

In particular, gESP achieves the information-theoretical bound $\Omega(k\log n)$ when $s(\ceiling{\sqrt{k}}) = \Theta(1)$, while the state-of-the-art result
\begin{equation}\label{sota}
      m = \Omega\left( ks(1) \log n  \right)
\end{equation}
requires $s(1) = \Theta(1)$, which is a stricter condition. The improvement results from considering $p$ entries in {\bf Step 1}, thus relaxing the condition on $x_{(1)}$ to the set $\{x_{(1)}, \dots, x_{(p)}\}$. Fig.~\ref{fig:illustration} provides a comprehensive illustration.

\begin{figure*}[t]
      \centering
      \includegraphics[width=0.9\linewidth]{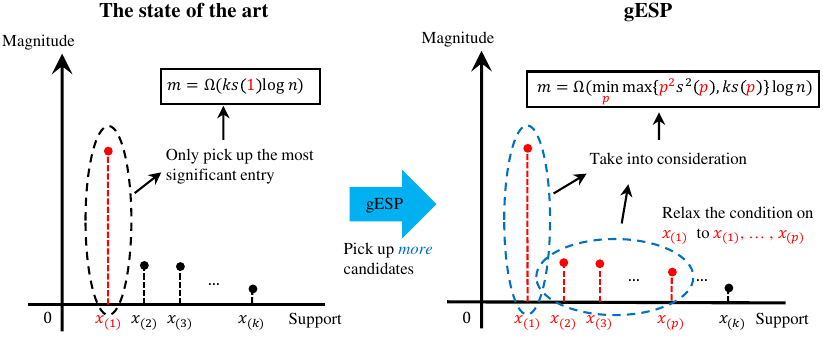}
      \caption{An illustration of the relaxed condition for achieving the optimal sample complexity $\Omega(k\log n)$. While state-of-the-art methods require the signal's energy to be concentrated in its single largest entry: $s(1) = \Theta(1)$, gESP achieves optimality for a broader class of signals where the energy is concentrated among the top $\ceiling{\sqrt{k}}$ entries: $s(\ceiling{\sqrt{k}})= \Theta(1)$.}
      \label{fig:illustration}
\end{figure*}

Moreover, even when $s(p)$ is unknown, our results (Corollary~\ref{coro:2} and~\ref{coro:3}) can still perform better in many cases. In the following, we set $p=k$ and consider Corollary~\ref{coro:2} for a simple discussion,
although we can prove the strict superiority of the result derived in Corollary~\ref{coro:3} (see Section~\ref{sec:DissB}). The result in Corollary~\ref{coro:2} is
\begin{equation}\label{dis-gesp}
      m = \Omega \left( k s^2(\lceil{\sqrt{k}}\rceil) \log n \right),
\end{equation}
Consequently,~\eqref{dis-gesp} outperforms the state-of-the-art result~\eqref{sota} once $\Theta(s^2(\lceil{\sqrt{k}}\rceil)) < \Theta(s(1))$. To illustrate this, we construct a parametric family of signals as follows.

\begin{example}
      \label{ex:intermediate-spike}
      Consider a family of $k$-sparse signals $\x\in\mathbb{R}^n$ parameterized by a decay exponent $\delta\in(1/4,1/2)$. Let
      \begin{equation}\label{eq:intermediate-spike-signal}
            x_{(i)}=
            \begin{cases}
                  k^{-\delta},            & i=1,                             \\[0.2em]
                  \frac{1}{2}k^{-\delta}, & 2\le i\le \lceil \sqrt{k}\rceil, \\[0.2em]
                  k^{-1/2},               & \lceil \sqrt{k}\rceil<i\le k,    \\[0.2em]
                  0,                      & i>k.
            \end{cases}
      \end{equation}
      For this family of signals, it holds that $s(1)=\Theta(k^{2\delta})$ whereas $s^2(\lceil\sqrt{k}\rceil)=\Theta(k^{4\delta-1})$, and thus we have $\Theta(s^2(\lceil{\sqrt{k}}\rceil)) < \Theta(s(1))$ with the improvement factor being $\Theta(k^{1-2\delta})$.
\end{example}

This family represents signals where most of the signal's energy is spread fairly evenly across a long ``flat tail'' of many small nonzero entries, while only a small fraction sits in the first $\sqrt{k}$ larger entries, so it looks like a modest spike plus a short plateau on top of a broad, energy-dominating background. The proof for this example is provided in Appendix~\ref{sec:proof-example-intermediate-spike}. Interestingly, when the condition $s(\lceil{\sqrt{k}}\rceil) = \Theta(1)$ holds, our result~\eqref{dis-gesp} also achieves the information-theoretic bound~\eqref{eq:klogn}.
Notably, even without prior knowledge of the signal structure $s(p)$, this optimality is attainable for the same class of signals as in the case where $s(p)$ is known, simply by setting the parameter $p=k$.

\subsection{Strict Superiority of Corollary \ref{coro:3} over Corollary \ref{coro:2}}\label{sec:DissB}
In Corollary~\ref{coro:3}, we claim that the result
\begin{equation*}
      m = \Omega \left( \min_{p \in [\ceiling{\sqrt k}]} \max \left\{p^2 s^2(p) ,\sqrt{k}s^2(p),  ks(p)  \right\} \log n \right),
\end{equation*}
is better than~\eqref{dis-gesp} in Section~\ref{sec:3.2}. In this section, we provide a concrete example to illustrate this strict superiority.
\begin{example}\label{eg:2}
      Assume that both $\sqrt{k}$ and $\sqrt[4]{k}$ are integers\footnote{In fact, we can also construct such examples even though this assumption does not hold. However, omitting it will lead to unnecessary complications in the construction (e.g., involving terms like $\ceiling{\sqrt{k}}$).}. Construct the target signal $\x$ as
      \begin{align*}
            |x_{(i)}|^2=
            \begin{cases}
                  \frac{\|\x\|^2}{\sqrt[4]{k^3}},                                                                         & i = 1,2,\cdots, \sqrt[4]{k},         \\
                  \frac{\|\x\|^2}{\sqrt{k} - \sqrt[4]{k}} \times \left(\frac{1}{\sqrt[3]{k}} - \frac{1}{\sqrt{k}}\right), & i = \sqrt[4]{k}+1, \cdots, \sqrt{k}, \\
                  \frac{\|\x\|^2}{k - \sqrt{k}} \times \left(1 - \frac{1}{\sqrt[3]{k}}\right),                            & i = \sqrt{k}+1, \cdots, k.
            \end{cases}
      \end{align*}
      In this case, it can be calculated that the result in Corollary~\ref{coro:2} is
      \begin{equation*}
            \Omega \left( k s^2(\sqrt{k}) \log n \right) = \Omega( k^{5/3} \log n),
      \end{equation*}
      while that in Corollary~\ref{coro:3} is
      \begin{align*}
              & \Omega \left( \min_{p \in [\ceiling{\sqrt k}]} \max \left\{p^2 s^2(p) ,\sqrt{k}s^2(p),  ks(p)  \right\} \log n \right) \\
            = & ~ \Omega \left( k s(\sqrt[4]{k}) \log n \right) = \Omega( k^{3/2} \log n).
      \end{align*}
      This demonstrates the strict improvement offered by Corollary~\ref{coro:3}.
\end{example}

% !TEX root = 00Main.tex
\section{Proof}\label{sec:proof}
The proofs of Theorem~\ref{thm:1} and Corollary~\ref{coro:1}--\ref{coro:3} rely primarily on the analysis of each step of gESP. In the following, we establish several propositions corresponding to each step, respectively.
\begin{itemize}[leftmargin=1.5em, itemsep=0.2em, parsep=0.2em]
      \item {\bf Step 1:} Generate an index set $S^0$ corresponding to the largest $p$ diagonal elements of $\mathbf{Z}$.
            \begin{proposition}\label{prop:1}
                  Provided the number of samples satisfies $m \geq Cp^2 s^2(p) \log n$, it holds with probability exceeding $1 - 2 n^{-c}$ that
                  \begin{equation}\label{propeq:1}
                        \frac{\|\x_{S^0}\|^2}{\|\x\|^2} \geq \frac{1}{2 s(p)},
                  \end{equation}
                  where $c$ and $C$ are universal constants.
            \end{proposition}

            \hspace{2mm}

      \item {\bf Step 2:} Set $\mathbf{e}^0$ as the unit leading eigenvector of $\mathbf{Z}_{S^0}$.
            \begin{proposition}\label{prop:2}
                  Suppose the set $S^0$ contains $p'$ indices and satisfies~\eqref{propeq:1}. Then, the inequality
                  \begin{equation}\label{propeq:2}
                        |\x^{*} \e^0| \geq \frac{1}{2} \|\x_{S^0}\|
                  \end{equation}
                  holds with probability at least $1-n^{-c}$ provided that $m \geq {C}{p's^2(p)} \log n$, where $c$ and $C$ are constants.
            \end{proposition}
            \begin{remark}
                  It may seem unusual to introduce a new parameter $p'$ in this proposition, given that $S^0$ contains $p$ indices in Algorithm~\ref{alg:gESP}. However, the proposition holds even when $p' = p$, serving as a generalization. We emphasize that this generalization will be useful in our proofs of Corollary~\ref{coro:2} and~\ref{coro:3}, where we do not have access to the optimal choice of $p$.

            \end{remark}

      \item {\bf Step 3:} Generate an index set $S^{1}$ corresponding to the largest $k$  entries of the vector $\f = \mathbf{Z} \mathbf{e}^0$ in magnitude.
            \begin{proposition}\label{prop:3}
                  Define the set $S_{\gamma}$ as
                  \begin{equation*}
                        S_{\gamma} \doteq\left\{j \in \text{\rm supp}(\mathbf{x}) \left| |x_j| \geq \frac{\gamma \|\mathbf{x}\|}{2\sqrt{k}}\right. \right\},
                  \end{equation*}
                  where $\gamma \in (0,1)$ is a constant. Moreover, suppose that \eqref{propeq:1} holds and $\e^0$ satisfies~\eqref{propeq:2}. Then, if the number of samples satisfies $m \geq  C \gamma^{-2}{k s(p)} \log n$, it holds with probability at least $1-n^{-c}$ that
                  $$S^{1} \supseteq S_{\gamma},$$
                  where $\mu_1$ is the hyperparameter given in~\eqref{propeq:2} and $c_{\mu_1,\gamma}$ is the constant depending on $\mu_1$ and $\gamma$. In other words, gESP selects all the indices of $S_{\gamma}$ in \textbf{\textup{Step 3}}. Furthermore, we have
                  \begin{equation}
                        \frac{\|\x_{S^{1}}\|^2}{\|\x\|^2} \geq 1 - \sum_{i \notin S_{\gamma}} \frac{|x_i|^2}{\|\x\|^2} \geq 1 - \frac{\gamma^2}{4}.
                  \end{equation}
            \end{proposition}
            Crucially, if $|x_{(k)}|$ is extremely small (violating~\eqref{eq:xmin-assumption}), it simply falls outside $S_{\gamma}$. This does not invalidate the proof. Instead, the “missed” small entries are absorbed into the approximation error term in \Cref{prop:4}. Since their magnitudes are small, the total error remains within the $\delta$-neighborhood ($\dist(\z,\x) \le \delta \|\x\|$), thus removing the need for a hard lower bound on $x_{(k)}$.

      \item {\bf Step 4:} Set $\z$ as the leading eigenvector of $\mathbf{Z}_{S^{1}}$ with $\|\z\|^2 = \lambda^2 \doteq \frac{1}{m} \sum_{i=1}^{m} y_i^2$.
            \begin{proposition}\label{prop:4}
                  Suppose $S^{1} \supseteq S_{\gamma}$. For any constant $0 < \delta < 1$,
                  \textbf{\textup{Step 4}} produces a signal estimate $\z$ falling into the $\delta$-neighborhood of $\x$, i.e.,
                  \begin{equation}\label{propeq:4}
                        \dist(\z,\x) \leq \delta \|\x\|
                  \end{equation}
                  with probability exceeding $1 - n^{-c}$ provided that $m \geq C \delta^{-2}k \log n$, where $c$ and $C$ are numerical constants.
            \end{proposition}

\end{itemize}

We defer the proofs of Propositions~\ref{prop:1}--\ref{prop:4} to the Appendix.

Now we proceed to prove the theorems in Section~\ref{sec:3}.
\begin{figure}[t]
      \centering
      \includegraphics[width=0.8\linewidth]{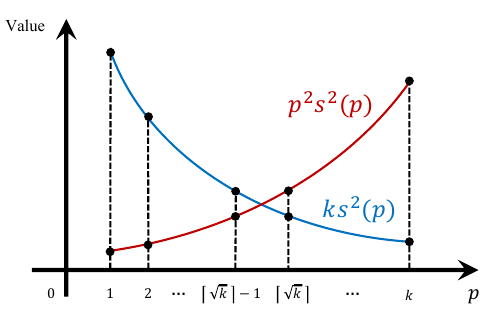}
      \caption{An Illustration of the two functions in \eqref{eq:gESP_k_result_2}.}
      \label{fig:complexityminmax}
\end{figure}
\begin{proof}[Proof of Theorem~\ref{thm:1}]
      Combining the results in Proposition~\ref{prop:1}--\ref{prop:4}  with $p'=p$, we can show that the sample complexity is
      \begin{equation}
            m = \Omega\left(\max\left\{p^2s^2(p), ks(p)\right\}\log n\right),
      \end{equation}
      and the overall probability is $1-n^{-c}$.
      Since we take the optimal value for $p$ as
      \begin{equation*}
            p_{\text{opt}} ={\arg \min}_{p \in [k]} \max \left\{p^2 s^2(p), ks(p) \right\},
      \end{equation*}
      we obtain the desired results.
\end{proof}

\begin{proof}[Proof of Corollary~\ref{coro:1}]
      For each fixed $p\in[k]$, combining the results in Proposition~\ref{prop:1}--\ref{prop:4} with $p'=p$, we can show that the final estimate $\z$ falls into the $\delta$-neighborhood of $\x$ with probability exceeding $1-n^{-c}$ when
      $$
            m = \Omega\left(\max\left\{p^2s^2(p), ks(p)\right\}\log n\right).
      $$
      To ensure that at least one estimate $\z$ satisfies $\dist(\z,\x) \leq \delta \|\x\|$, we can apply the union bound for all $p\in [k]$. Then the sample complexity becomes
      $$
            m = \Omega \left( \min_{p \in [k]} \max \left\{p^2 s^2(p) \log n, ks(p) \log n \right\} \right),
      $$
      and the probability is at least $1-kn^{-c}$. This is at least $1-n^{-c}$ if we take a sufficiently large constant $c$ since $n>k$.
      Thus the proof is complete.
\end{proof}

\begin{figure*}[t]
      \centering
      \subfloat[Gaussian signal.  $k$ = 10.\label{simu:1.1}]{
            \includegraphics[width=0.31\textwidth]{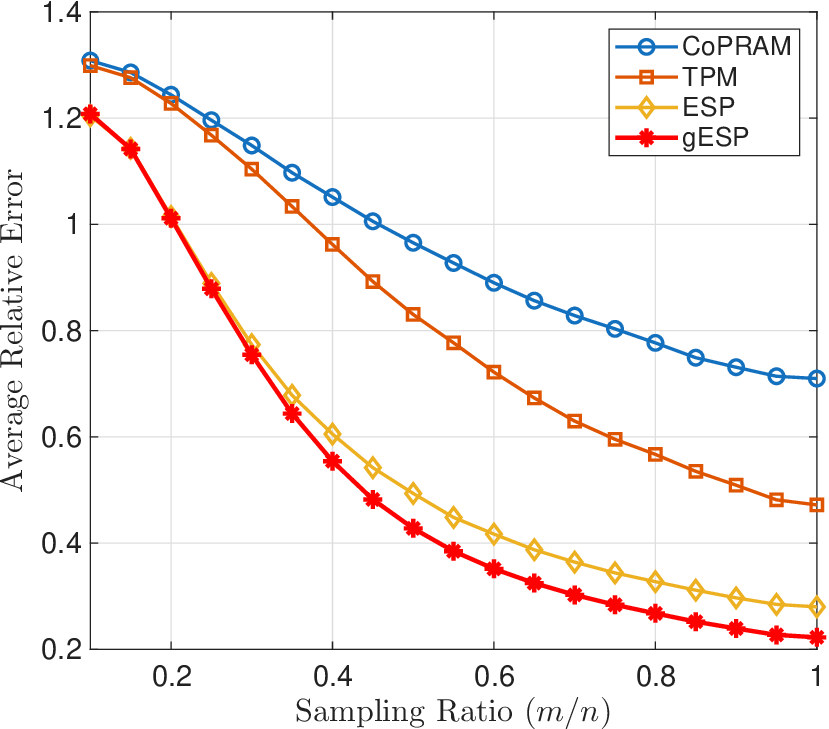}
      }\hfill
      \subfloat[$0$-$1$ signal. $k$ = 10. \label{simu:1.2}]{
            \includegraphics[width=0.31\textwidth]{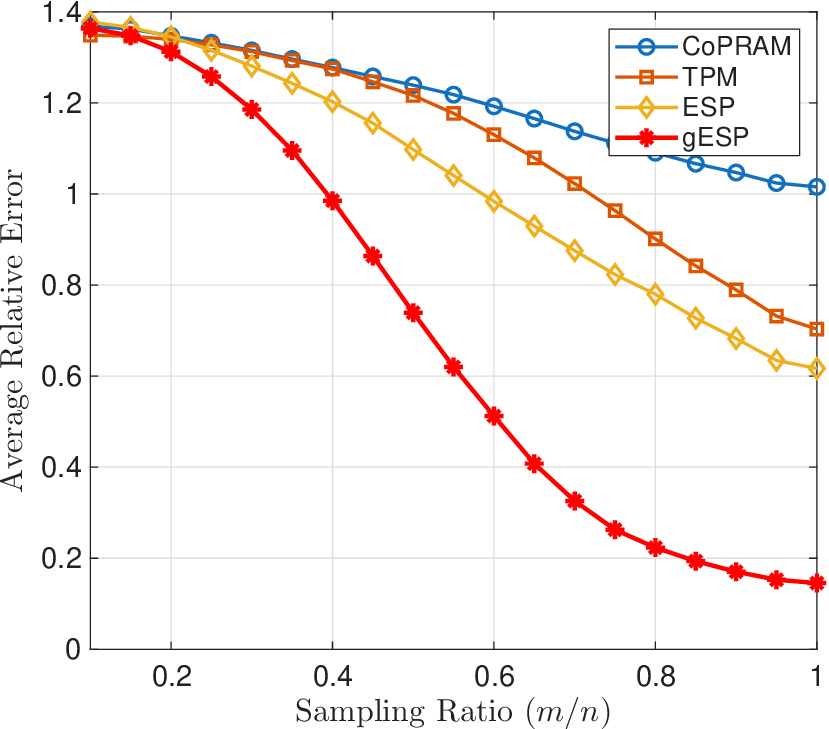}
      }
      \hfill
      \subfloat[Exponential decaying signal. $k$ = 10.\label{simu:1.5}]{
            \includegraphics[width=0.31\textwidth]{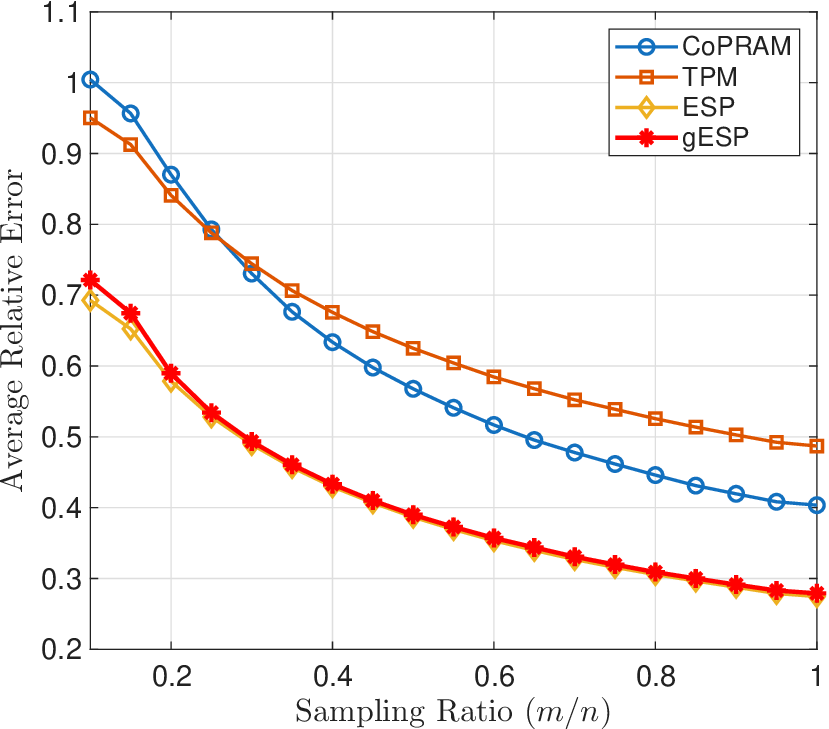}
      }
      \caption{Performance comparison of relative error and fraction of recovered support as a function of sample ratio.}
\end{figure*}
\begin{proof}[Proof of Corollary~\ref{coro:2}]
      For any $p\in[k]$, denote $S'$ as the set corresponding to the largest $p$ diagonal elements of $\mathbf{Z}$. Note that $S^0$ corresponds to the largest $k$ diagonal elements of $\mathbf{Z}$.
      It holds that
      \begin{equation}
            \frac{\|\x_{S^0}\|^2}{\|\x\|^2} \overset{(a)}{\geq} \frac{\|\x_{S'}\|^2}{\|\x\|^2} \overset{(b)}{\geq} \frac{1}{2 s(p)},
      \end{equation}
      where $(a)$ is due to $k\geq p$, and $(b)$ holds by applying Proposition~\ref{prop:1} to the set $S'$ when $m$ $\geq Cp^2 s^2(p) \log n$.

      Then, from the results above and in Proposition~\ref{prop:2}--\ref{prop:4} (with $p'=k$), it holds that for any $p\in[k]$, when
      \begin{equation}\label{eq:thm3proof1}
            m = \Omega \left(\max \left\{p^2 s^2(p) \log n, ks^2(p) \log n \right\} \right),
      \end{equation}
      the final estimate $\z$ falls into the $\delta$-neighborhood of $\x$ with probability exceeding $1-n^{-c}$. This result also holds for all $p\in[k]$ since we can take the union bound, and the probability will be $1-kn^{-c}$. It is actually $1-n^{-c}$ if we take a sufficiently large constant $c$ since $n>k$.

      Note that the choice of $p\in[k]$ is arbitrary. Therefore, we can choose the best $p$  to make the sample complexity \eqref{eq:thm3proof1} smallest, which becomes
      \begin{equation}\label{eq:gESP_k_result_2}
            m = \Omega \left(\min_{p\in [k]}\max \left\{p^2 s^2(p) \log n, ks^2(p) \log n \right\} \right),
      \end{equation}
      It remains to show that \eqref{eq:gESP_k_result_2} is equivalent to \eqref{eq:gESP_k_result}.

      As stated following Definition~\ref{def:s}, $p s(p)$ is monotonically increasing, and $s(p)$ is monotonically decreasing. Both functions are positive. Accordingly, we can illustrate the behavior of the two functions in \eqref{eq:gESP_k_result_2}; see Fig.~\ref{fig:complexityminmax}.

      Therefore, the minimum occurs when
      \begin{equation}
            p=\lceil{\sqrt{k}}\rceil~~~~ \text{or} ~~~~p=\lceil{\sqrt{k}}\rceil - 1.
      \end{equation}
      From Lemma~\ref{lemma:ceiling}, we know that both cases imply the same sample complexity
      \begin{equation}
            m = \Omega \left( k s^2(\lceil{\sqrt{k}}\rceil) \log n \right),
      \end{equation}
      which completes the proof.

\end{proof}

\begin{proof}[Proof of Corollary~\ref{coro:3}]
      The proof is analogous to the proof of Corollary~\ref{coro:2}. Assume any fixed $p \in [\lceil{\sqrt{k}}\rceil]$. The sample complexity for {\bf Step 1} is still $\Omega(p^2s^2(p)\log n)$, and applying Proposition~\ref{prop:2}--\ref{prop:4} on the subsequent steps with $p'=\lceil{\sqrt{k}}\rceil$  finally leads to
      \begin{equation}
            m = \Omega \left(\max \left\{p^2 s^2(p) ,\sqrt{k}s^2(p),  ks(p)  \right\} \log n \right).
      \end{equation}
      Since the choice of $p\in[k]$ is arbitrary, taking the union bound yields the desired result
      \begin{equation*}
            m = \Omega \left( \min_{p \in [\ceiling{\sqrt k}]} \max \left\{p^2 s^2(p) ,\sqrt{k}s^2(p),  ks(p)  \right\} \log n \right).
      \end{equation*}

\end{proof}

% !TEX root = 00Main.tex
\section{Numerical Simulations}\label{sec:simu}
In this section, we present numerical experiments to evaluate the performance of our proposed gESP algorithm against several state-of-the-art methods.

\subsection{Experimental Setup}
In all simulations, the signal dimension is fixed at $n = 1000$. The measurement vectors $\{\mathbf{a}_i\}_{i=1}^{m} \subset \mathbb{C}^n$ are {\it i.i.d.} standard complex Gaussian vectors. The $k$-sparse signal $\mathbf{x} \in \mathbb{C}^n$ is constructed by first selecting a support of size $k$ uniformly at random. The non-zero entries are then generated from one of three models: i) {\it i.i.d.} standard complex Gaussian, ii) all ones (binary on support), and iii) exponentially decaying magnitudes. All reported results are averaged over $1000$ independent trials to mitigate random fluctuations.

We compare gESP with CoPRAM, TPM, and ESP. For a fair comparison, we focus on the initialization stage of CoPRAM, and use the recommended hyperparameters for TPM. For gESP, we set the theoretically optimal parameter $p_{opt}$ from~\eqref{eq:pstar} for each trial.

The performance metric is the relative error, defined as:
$$
      \text{Relative Error} = \frac{\dist(\z, \x)}{\|\mathbf{x}\|_2},
$$
where $\mathbf{z}$ is the estimated signal and $\mathbf{x}$ is the true signal~\cite{CJF,TAF}.

\subsection{Performance Comparison}
In this experiment, we vary the sample ratio $m/n$ from $0.05$ to $1.0$ and record the corresponding relative error.

As depicted in Figs.~\ref{simu:1.1}--\ref{simu:1.5}, the performance on sparse Gaussian signals provides a direct and strong validation of our theoretical analysis. For this signal class, gESP was configured using the parameter $p$ determined by our theoretically derived optimum. Consistent with the predictions of our theory, gESP achieves the best performance, consistently outperforming ESP and other state-of-the-art methods, albeit slightly. This clear correspondence between our theory and the empirical results underscores the accuracy of our analysis for signals with rich amplitude variations.

In contrast, the results for binary and exponentially decaying signals, while not perfectly aligned with asymptotic predictions, offer valuable insights into the practical behavior of the algorithms. For binary signals, the theory suggests that the complexity remains $\Omega(k^2 \log n)$ regardless of the value of $p$, leading one to expect similar performance across methods. However, we set $p=k$ in gESP, and it demonstrates a surprisingly large performance margin. We attribute this significant practical advantage to the influence of nonasymptotic constant factors that our analysis does not model but which clearly favor the gESP framework. Similarly, for exponentially decaying signals, gESP is marginally outperformed by ESP. This is also likely an effect of such constants; our theory provides the optimal scaling for $p$, but in this finite-dimensional regime, the true optimal value is empirically found to be $p=1$ (i.e., ESP). This does not contradict our theory but rather highlights that while it provides a powerful general guideline, the precise optimal parameter in practice can be influenced by the signal's specific structure.

\subsection{Robustness to the Choice of Parameter $p$}
\label{sec:exp_robustness}

A core feature of gESP lies in that it uses the tuning parameter $p$ to balances the captured signal energy against the accumulated perturbation. In this subsection, we experimentally investigate the sensitivity of gESP to the choice of $p$ and demonstrate its practical performance.

We consider a power-law signal with decay rate $\alpha=0.5$ (i.e., $|x_{(i)}| \propto i^{-0.5}$ for $i=1,\ldots,k$). This structure of signal is particularly selected as it is a representative ``heavy-tailed'' case: neither essentially single-spiked nor perfectly flat. Such a signal profile effectively challenges the algorithm to balance energy and perturbation, making it an ideal candidate to evaluate the impact of different selection strategies. In the following, we test the performance of gESP with different choices of $p$ on this power-law signal:
\begin{enumerate}
      \item The standard single-peak approach $p=1$ (corresponding to ESP).
      \item The theoretical oracle $p$ obtained by minimizing the bound in \Cref{thm:1}.
      \item The heuristic surrogate $p = \lceil\sqrt{k}\rceil$ as a middle-ground choice.
      \item The full support selection $p=k$.
\end{enumerate}

\begin{figure}[htbp]
      \centering
      \includegraphics[width=0.45\textwidth]{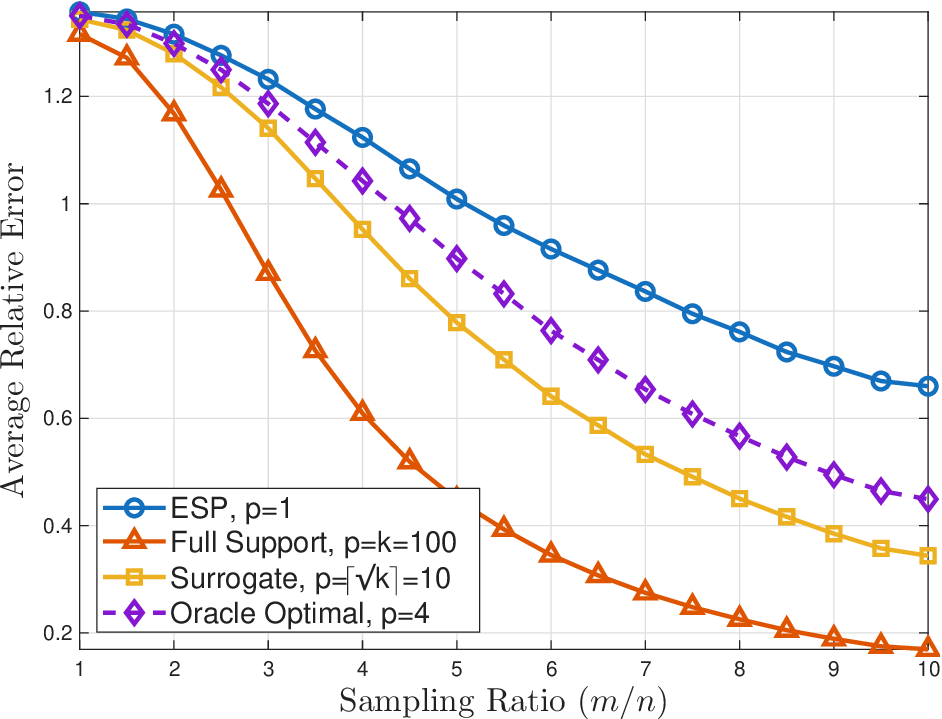}
      \caption{Performance comparison of gESP with different choices of $p$ under the designed power-law signal. We set $n=1000, k=100$ with $m/n$ varying from $0.05$ to $1.0$. The multi-peak strategies ($p > 1$) significantly outperform the single-peak baseline. Notably, the algorithm is robust to large $p$, with the full support choice ($p=k$) achieving the best empirical performance.}
      \label{fig:robustness}
\end{figure}

The results are shown in \Cref{fig:robustness}. We observe that all multi-peak strategies ($p > 1$) significantly outperform the baseline ESP ($p=1$). Notably, the empirical performance remains stable or even improves as $p$ increases, with the full support choice ($p=k$) achieving the best performance in this setting.

This discrepancy between the theoretical oracle and the empirical optimum can be attributed to two main factors. First, our theoretical analysis is asymptotic in nature (assuming $m, n, k \to \infty$), whereas the experiments are conducted in a finite-dimensional setting ($n=1000, k=100$). Clearly the finite-sample behavior may arguably tolerate more perturbation than the asymptotic worst-case prediction. Second, the ``oracle'' choice is derived by minimizing a theoretical upper bound on the sample complexity. This bound relies on standard concentration inequalities (e.g., Bernstein-type bounds), which are sufficient yet inherently conservative. As a result, the theoretical derivation tends to overestimate the perturbation accumulation associated with large $p$, leading to a more conservative recommendation compared to the algorithm's actual robustness behavior in practice.

Practically, this result is encouraging: it implies that users can simply select a sufficiently large $p$ (e.g., above the expected sparsity level) to achieve the near-optimal performance without precise knowledge of the energy profile $s(p)$.

% !TEX root = 00Main.tex
\section{Conclusion}\label{sec:conc}

In this paper, we have addressed the problem of sparse phase retrieval, focusing on improving initialization algorithms for recovering a $k$-sparse $n$-dimensional signal from $m$ phaseless observations. Motivated by the limitations of existing methods that rely on specific structural assumptions of signals, we have proposed generalized Exponential Spectral Pursuit, which promotes the initialization performance and also enhances the understanding of sample complexity in sparse phase retrieval. Empirical evaluations further validate the practical effectiveness of gESP. Simulation experiments demonstrate its robustness and efficiency in recovering sparse signals, highlighting its potential utility in real-world applications.

Although this work advances both theoretical bounds and practical performance for sparse phase retrieval, it does not fully resolve the statistical-to-computational gap. Specifically, under Gaussian measurements, an information-theoretic bound of $\Omega(k\log n)$ samples is sufficient to uniquely determine the target signal. However, for polynomial-time recovery, the sample complexity increases to $\Omega(k^2\log n)$ in the worst case, for example, when the target signal consists of binary entries. Future research could aim to bridge this gap by exploring novel algorithmic frameworks or leveraging insights from other areas, such as computational complexity and optimization. Furthermore, extending the applicability of gESP to broader measurement models and optimizing the subsequent refinement stage offer promising directions to advance this field.

% !TEX root = 00Main.tex
\appendices
\numberwithin{equation}{section}
\newcounter{mytempthcnt}
\setcounter{mytempthcnt}{\value{theorem}}
\setcounter{theorem}{\value{mytempthcnt}}

\section{Basic Tools and Lemmas} \label{app:pre}

In this section, we present several useful lemmas and necessary proofs.
Recall the definition of $\mathbf{Z}$:
\begin{equation}\label{eq:Z}
      \mathbf{Z}=\frac{1}{m} \sum_{i=1}^{m} \left(\frac{1}{2} - \exp\left(- \frac{y_i^2}{\lambda^2} \right) \right) \mathbf{a}_i \mathbf{a}_{i}^{*},
\end{equation}
where $\lambda^2 = \frac{1}{m}\sum_{i=1}^{m}y_i^2$. Deriving the exact expectation of $\mathbf{Z}$ is challenging due to the dependency between $y_i$ and $\lambda^2$. Following the technique in \cite{XuZQ, ESP}, we introduce a proxy matrix:
\begin{equation}\label{eq:Zt}
      \tilde{\mathbf{Z}}=\frac{1}{m} \sum_{i=1}^{m} \left(\frac{1}{2} - \exp\left(- \frac{y_i^2}{\|\x\|^2} \right) \right) \mathbf{a}_i \mathbf{a}_{i}^{*},
\end{equation}
whose expectation was derived in \cite{XuZQ} as
\begin{equation}\label{eq:EtZ}
      \mathbb{E}[\tilde{\mathbf{Z}}] = \frac{\x\x^*}{4\|\x\|^2}.
\end{equation}

For convenience, denote
\begin{equation}\label{eq:def-Zf}
      \begin{aligned}
            Z_j        & = \diag(\mathbf{Z})_j,   & \tilde{Z}_j        & = \diag(\tilde{\mathbf{Z}})_j,    \\
            \mathbf{f} & =\mathbf{Z}\mathbf{e}^0, & \tilde{\mathbf{f}} & = \tilde{\mathbf{Z}}\mathbf{e}^0.
      \end{aligned}
\end{equation}
In the following, we introduce several results concerning the effect of the ceiling function on the sample complexity.
\begin{lemma}\label{lemma:ceiling}
      For the ceiling function in the asymptotic notation, we have
      \begin{align}
            \Omega(\ceiling{\sqrt{k}})    & = \Omega(\sqrt{k}), \label{ceil:1}               \\
            \Omega(s(\ceiling{\sqrt{k}})) & = \Omega(s(c\ceiling{\sqrt{k}})), \label{ceil:2}
      \end{align}
      where $c$ is a fixed positive integer.
\end{lemma}
\begin{proof}
      Note that $\sqrt{k} \leq \ceiling{\sqrt{k}} \leq \sqrt{k} + 1$, then~\eqref{ceil:1} can be easily derived from the limit:
      \begin{equation*}
            \lim_{k \to \infty} \frac{\sqrt{k}+1}{\sqrt{k}} = 1.
      \end{equation*}
      As for~\eqref{ceil:2}, recall from the definition of $s(p)$, we have
      \begin{equation*}
            s(c\ceiling{\sqrt{k}}) = \frac{\|\x\|^2}{\sum_{j=1}^{c \ceiling{\sqrt{k}}} |x_{(j)}|^2}.
      \end{equation*}
      Since $x_{(j)}$'s are the rearrangement of the entries of $\x$ in descending order of their magnitudes, the summation in the denominator can be estimated as
      \begin{equation*}
            \sum_{j=1}^{\ceiling{\sqrt{k}}} |x_{(j)}|^2 \leq \sum_{j=1}^{c \ceiling{\sqrt{k}}} |x_{(j)}|^2 \leq c \sum_{j=1}^{\ceiling{\sqrt{k}}} |x_{(j)}|^2,
      \end{equation*}
      which in turn yields
      \begin{equation*}
            \frac{1}{c} s(\ceiling{\sqrt{k}}) \leq s(c\ceiling{\sqrt{k}}) \leq s(\ceiling{\sqrt{k}}).
      \end{equation*}
      Since $c$ is a constant as $k$ and $n$ tend to infinity,~\eqref{ceil:2} follows naturally.
\end{proof}

The next lemma gives the error bound of \(\lambda^2\).
\begin{lemma}\label{lemma:x-lambda-diff}
      Consider the problem \eqref{PR}. With probability at least $1 - n^{-c}$, it holds that
      \begin{equation*}
            \left| \lambda^2 - \|\x\|^2 \right| \le C \sqrt{\frac{\log n}{m}} \|\x\|^2.
      \end{equation*}
\end{lemma}

\begin{proof}
      Let $y_i = \a_i^* \x$. Since $\a_i \sim \mathcal{CN}(0, \mathbf{I})$, the variable $y_i \sim \mathcal{CN}(0, \|\x\|^2)$. The term $|\a_i^* \x|^2 / \|\x\|^2$ follows a standard exponential distribution, which is sub-exponential with constant norm.

      We apply the standard Bernstein inequality for the sum of independent sub-exponential random variables \cite[Theorem 2.8.1]{hdpBook}. For any $t \ge 0$,
      \begin{align*}
             & \mathbb{P}\left( \left| \frac{1}{m} \sum_{i=1}^m |\a_i^*\x|^2 - \|\x\|^2 \right| \ge t \|\x\|^2 \right) \\
             & ~~~~~~~~~~\le 2 \exp\left( -c m \min(t^2, t) \right).
      \end{align*}
      Setting $t = C\sqrt{\frac{\log n}{m}}$ yields the desired result.
\end{proof}

The following lemma gives the operator norm bound of the weighted average of random matrices, which is a natural extension of the classical result in random matrix theory.
\begin{lemma}[Uniform Concentration with Bounded Weights]\label{lemma:weighted-concentration}
      Let $\mathbf{a}_i \sim \mathcal{CN}(0, \mathbf{I}_n)$ be independent standard complex Gaussian vectors. Let $\xi_i \in \mathbb{R}$ be independent scalar random variables (potentially depending on $\mathbf{a}_i$) satisfying $|\xi_i| \le B$ almost surely for some constant $B > 0$.

      With probability at least $1 - 2\exp(-cs\log(n/s))$, the following holds uniformly for all subsets $S \subset [n]$ with cardinality $|S| = s$, provided the sample size satisfies $m \geq C s\log(n/s)$:
      \begin{align}\nonumber
             & \left\| \frac{1}{m} \sum_{i=1}^m \xi_i (\mathbf{a}_i)_S (\mathbf{a}_i)_S^* - \mathbb{E}\left[ \xi_1 (\mathbf{a}_1)_S (\mathbf{a}_1)_S^* \right] \right\| \\ \le ~ & C' B \sqrt{\frac{s \log(n/s)}{m}},\label{eq:weighted-op}
      \end{align}
      where $c, C, C'$ are universal numerical constants.
\end{lemma}

\begin{proof}
      The proof follows the standard covering argument used for covariance matrix estimation (see, e.g., the proof of \cite[Theorem 4.6.1]{hdpBook}). The basic intuition is that the weighted bounded scalar will not essentially change the sub-exponential property of the random matrix. For simplicity, we denote
      \[
            \mathbf{E}_i := \xi_i (\mathbf{a}_i)_S (\mathbf{a}_i)_S^* - \mathbb{E}\left[ \xi_1 (\mathbf{a}_1)_S (\mathbf{a}_1)_S^* \right],
      \]
      and \(\mathbf{E} := \frac{1}{m}\sum_{i=1}^m \mathbf{E}_i\).

      \textbf{Step 1: Discretization via $\epsilon$-net.}
      Fix a support $S$ with $|S| = s$. Let $\mathcal{N}$ be a $\frac{1}{4}$-net of the unit sphere in $\mathbb{C}^{|S|}$. The cardinality of the net is bounded by $|\mathcal{N}| \le 9^{2|S|} \le 9^{4s}$.
      From \cite[Lemma 4.4.1]{hdpBook}, the operator norm of a Hermitian matrix $\mathbf{E}$ is controlled by the net
      \begin{equation}\label{eq:net-control}
            \|\mathbf{E}\| \le 2 \sup_{u \in \mathcal{N}} |\u^* \mathbf{E} \u|.
      \end{equation}
      Thus, it suffices to bound the scalar deviation for each fixed $\u \in \mathcal{N}$.

      \textbf{Step 2: Scalar Concentration.}
      Fix a unit vector $\u \in \mathcal{N}$. Consider the scalar random variable
      \[
            Z_i = \u^* \left( \xi_i (\mathbf{a}_i)_S (\mathbf{a}_i)_S^* \right) \u = \xi_i |\langle (\mathbf{a}_i)_S, \u \rangle|^2.
      \]
      Since $(\mathbf{a}_i)_S$ is standard Gaussian, the inner product $g = \langle (\mathbf{a}_i)_S, \u \rangle$ is a standard scalar complex Gaussian variable. The term $|g|^2$ follows an exponential distribution, which is sub-exponential with norm $\||g|^2\|_{\psi_1} \le C_{sub}$.
      Because the weights are bounded ($|\xi_i| \le B$), the variable $Z_i$ is sub-exponential with norm
      \[
            \|Z_i\|_{\psi_1} = \|\xi_i |g|^2\|_{\psi_1} \le B \||g|^2\|_{\psi_1} \le C_{sub} B.
      \]
      We now apply the scalar Bernstein inequality for sub-exponential variables \cite[Theorem 2.8.1]{hdpBook}  and obtain
      \begin{align*}
             & \mathbb{P}\left( \left| \frac{1}{m} \sum_{i=1}^m (Z_i - \mathbb{E}Z_i) \right| \ge t \right) \\
             & ~~~~~~~~\le     2 \exp\left( -c m \min\left( \frac{t^2}{K^2}, \frac{t}{K} \right) \right),
      \end{align*}
      where $K = \max_i \|Z_i\|_{\psi_1} \lesssim B$.
      Assuming large $m$ such that the sub-gaussian tail dominates, we set $t = C' B \sqrt{\frac{s \log(n/s)}{m}}$ and get
      \begin{align*}
            \mathbb{P}\Big( \Big| \underbrace{\frac{1}{m} \sum_{i=1}^m (Z_i - \mathbb{E}Z_i)}_{\u^*\mathbf{E}_S\u} \Big| \ge t \Big)
            \le~ 2 \exp(- c' s \log(n/s)).
      \end{align*}
      Note that the $B^2$ in the variance term cancels out, leaving the rate dependent on $m$ and $s$.

      \textbf{Step 3: Union Bound.}
      We now control the probability that the operator norm exceeds $t$ for any valid support $S$. Taking the union bound over all $\binom{n}{s}$ supports and all $\u \in \mathcal{N}$ for each support yields
      \begin{align*}
            \mathbb{P}\left( \sup_{|S|= s} \|\mathbf{E}_S\| \ge t \right)
             & \eqmakebox[op]{$\overset{\eqref{eq:net-control}}{\le}$} \sum_{|S|= s} \sum_{\u \in \mathcal{N}} \mathbb{P}\left( |\u^* \mathbf{E}_S \u| \ge t/2 \right) \\
             & \eqmakebox[op]{$\le$} \binom{n}{s} \cdot 9^{4s} \cdot 2 \exp\left( - c' s \log(n/s) \right)                                                             \\
             & \eqmakebox[op]{$\le$} \left( \frac{en}{s} \right)^{s} 9^{4s} \cdot 2 \exp\left( - c' s \log(n/s) \right).
      \end{align*}
      By choosing the constant $C$ in the sample complexity $m \ge C s \log(n/s)$ sufficiently large, the exponent $-c' s \log(n/s)$ dominates and thus ensures that the total probability is at most $2\exp(-c s \log(n/s))$.
\end{proof}

Based on the above lemmas, we proceed to show the similarity between the exponential spectral $\mathbf{Z}$ and its estimator $\mathbb{E}[\tilde{\mathbf{Z}}]$. First, we bound the operator norm of the difference matrix $\tilde{\mathbf{Z}} - \mathbb{E}[\tilde{\mathbf{Z}}]$. In fact, this is a direct application of \Cref{lemma:weighted-concentration}.
\begin{lemma}\label{lemma:op-diff1}
      Let $\tilde{\mathbf{Z}}$ be defined as in~\eqref{eq:Zt} and its expectation is given by \eqref{eq:EtZ}. For all $S \subseteq [n]$ with $|S|= s$, the following bound holds for the operator norm with probability exceeding $1 - 2\exp(-cs\log(n/s))$:
      \begin{equation}\label{eq:op-diff1}
            \| \tilde{\mathbf{Z}}_S - \mathbb{E}[\tilde{\mathbf{Z}}_S] \| \le C\sqrt{\frac{s \log(n/s)}{m}}.
      \end{equation}
\end{lemma}
\begin{proof}
      Note that the scalar coefficients $\xi_i = \frac{1}{2} - \exp(-y_i^2/\|\x\|^2)$ are strictly bounded, since $|\xi_i| \le 1/2$. By applying \Cref{lemma:weighted-concentration}, we obtain the desired result.
\end{proof}

Next, we bound the operator norm of the difference matrix $\mathbf{Z} - \tilde{\mathbf{Z}}$.

\begin{lemma}\label{lemma:op-diff2}
      Let $\mathbf{Z}$ and $\tilde{\mathbf{Z}}$ be defined as in~\eqref{eq:Z} and~\eqref{eq:Zt}. For all $S \subseteq [n]$ with $|S| = s$, the following bound holds for the operator norm with probability exceeding $1 - 2\exp(-cs\log(n/s))$:
      \begin{equation}\label{eq:op-diff2}
            \| \mathbf{Z}_S - \tilde{\mathbf{Z}}_S \| \le C\sqrt{\frac{\log n}{m}}.
      \end{equation}
\end{lemma}

\begin{proof}
      Write the difference matrix $\Delta \mathbf{Z}_{S} = \mathbf{Z}_{S} - \tilde{\mathbf{Z}}_{S}$:
      \begin{equation}\label{eq:delta-z}
            \Delta \mathbf{Z}_{S} = \frac{1}{m} \sum_{i=1}^m \left[ \exp\left(-\frac{y_i^2}{\|\x\|^2}\right) - \exp\left(-\frac{y_i^2}{\lambda^2}\right) \right] (\a_i)_{S} (\a_i)_{S}^*.
      \end{equation}
      Define the scalar function $g(t) = \exp(-y_i^2/t)$. By the Lagrange Mean Value Theorem, there exists a $\xi_i$ between $\min(\lambda^2, \|\x\|^2)$ and $\max(\lambda^2, \|\x\|^2)$ such that
      \begin{equation*}
            g(\|\x\|^2) - g(\lambda^2) = g'(\xi_i) (\|\x\|^2 - \lambda^2).
      \end{equation*}
      The derivative is $g'(t) = \frac{y_i^2}{t^2} \exp(-y_i^2/t)$. Thus,
      \begin{equation*}
            \exp\left(-\frac{y_i^2}{\|\x\|^2}\right) - \exp\left(-\frac{y_i^2}{\lambda^2}\right) = \underbrace{\frac{y_i^2}{\xi_i^2} \exp\left(-\frac{y_i^2}{\xi_i}\right)}_{\delta_i\ge 0} (\|\x\|^2 - \lambda^2).
      \end{equation*}
      Substituting this back into the matrix sum yields
      \begin{equation*}
            \Delta \mathbf{Z}_S = (\|\x\|^2 - \lambda^2) \cdot \left( \frac{1}{m} \sum_{i=1}^m \delta_i (\a_i)_{S} (\a_i)_{S}^* \right).
      \end{equation*}
      Now, we apply the operator norm $\|\cdot\|$. Since $|\lambda^2 - \|\x\|^2|$ is a global scalar, we can pull it out and obtain
      \begin{equation}\label{eq:op-diff-z}
            \| \Delta \mathbf{Z}_S \| = |\|\x\|^2 - \lambda^2| \cdot \left\| \frac{1}{m} \sum_{i=1}^m \delta_i (\a_i)_{S} (\a_i)_{S}^* \right\|.
      \end{equation}

      First, we bound the scalar coefficient $\delta_i$ uniformly. Consider the function $h(u) = u e^{-u}$ for $u \ge 0$. Its maximum value is $e^{-1}$ at $u=1$.We can rewrite $\delta_i$ as
      \begin{equation*}
            \delta_i = \frac{1}{\xi_i} \left( \frac{y_i^2}{\xi_i} \exp\left(-\frac{y_i^2}{\xi_i}\right) \right) \le \frac{1}{\xi_i} \cdot e^{-1}.
      \end{equation*}
      Note that for each $i$, $\xi_i$ is between $\min(\lambda^2, \|\x\|^2)$ and $\max(\lambda^2, \|\x\|^2)$. From \Cref{lemma:x-lambda-diff}, denote $\epsilon = C\sqrt{\frac{\log n}{m}}$, and we have $\xi_i \ge (1-\epsilon)\|\x\|^2$. Thus,
      \begin{equation}\label{eq:delta-bound}
            0 \le \delta_i \le \frac{1}{e(1-\epsilon)\|\x\|^2} \doteq \frac{C_1}{\|\x\|^2}.
      \end{equation}

      Second, we bound $\left\| \frac{1}{m} \sum_{i=1}^m \delta_i (\a_i)_{S} (\a_i)_{S}^* \right\|$ via the PSD property. Since $(\a_i)_{S} (\a_i)_{S}^*$ are Positive Semidefinite (PSD) matrices and coefficients $\delta_i$ are non-negative, it holds that
      \begin{equation*}
            \left\| \frac{1}{m} \sum_{i=1}^m \delta_i (\a_i)_{S} (\a_i)_{S}^* \right\| \le \left( \max_i \delta_i \right) \left\| \frac{1}{m} \sum_{i=1}^m (\a_i)_{S} (\a_i)_{S}^* \right\|.
      \end{equation*}
      Using the bound~\eqref{eq:delta-bound} for $\delta_i$ and applying \Cref{lemma:weighted-concentration} with $\xi_i \equiv 1$, we have
      \begin{equation*}
            \left\| \frac{1}{m} \sum_{i=1}^m \delta_i (\a_i)_{S} (\a_i)_{S}^* \right\| \le \frac{C_1}{\|\x\|^2} \left(1+C_2\sqrt{\frac{s \log(n/s)}{m}}\right).
      \end{equation*}

      Finally, substituting this back into equation \eqref{eq:op-diff-z} yields
      \begin{equation*}
            \| \Delta \mathbf{Z}_S \| \le |\|\x\|^2-\lambda^2| \cdot \frac{C_1}{\|\x\|^2} \left(1+C_2\sqrt{\frac{s \log(n/s)}{m}}\right).
      \end{equation*}
      Applying \Cref{lemma:x-lambda-diff}, we have
      \begin{equation*}
            \| \Delta \mathbf{Z}_S \| \le C_1C_3\sqrt{\frac{\log n}{m}} \cdot \left(1+C_2\sqrt{\frac{s \log(n/s)}{m}}\right).
      \end{equation*}
      Choosing $m \geq C_4 s \log(n/s)$ with sufficiently large $C_4$ allows the union bound to control the failure probability by $2\exp(-cs\log(n/s))$. This completes the proof.
\end{proof}

As a consequence, via triangle inequality, we have the following result.
\begin{corollary}
      \label{cor:op-diff}
      For all $S \subseteq [n]$ with $|S| = s $, the following concentration inequality
      \begin{equation}\label{eq:op-diff}
            \| \mathbf{Z}_S - \mathbb{E}[\tilde{\mathbf{Z}}_S] \| \le C\sqrt{\frac{s \log(n/s)}{m}}
      \end{equation}
      holds true with probability exceeding $1 - 2e^{-cs\log(n/s)}$ given $m \geq Cs\log(n/s)$, where $c$ and $C$ are numerical constants.
\end{corollary}

\subsection{Concentration analysis for the sum of the diagonal entries}
\begin{lemma}\label{conc-uv}
      For all $S \subseteq [n]$ with $|S| = p$ and constant $0 < \eta < 1$, the following concentration inequality
      \begin{equation}\label{conc:0}
            \left \lvert \frac{1}{p}  \sum_{j \in S} \left( Z_{j} - \mathbb{E}\left[\tilde{Z}_{j}\right] \right) \right \rvert < \eta \times \frac{1}{ps(p)}
      \end{equation}
      holds true with probability exceeding $1 - n^{-c_{\eta}}$ given $m \geq C_{\eta} p^2 s^2(p) \log n$, where $c_{\eta}$ and $C_{\eta}$ are numerical constants related to $\eta$.

\end{lemma}
\begin{proof}
      The left term of~\eqref{conc:0} can be expressed as
      \begin{align*}
            \frac{1}{p}  \sum_{j \in S} \left( Z_{j} - \mathbb{E}\left[\tilde{Z}_{j}\right] \right) & = \frac{1}{p} \sum_{j \in S}  \left( \mathbf{Z}_{\{j\}} - \mathbb{E}[\tilde{\mathbf{Z}}_{\{j\}}] \right) \\
                                                                                                    & \le \max_{j \in S} \left\lvert \mathbf{Z}_{\{j\}} - \mathbb{E}[\tilde{\mathbf{Z}}_{\{j\}}] \right\rvert.
      \end{align*}
      For all $j \in [n]$, employing \Cref{cor:op-diff} with $S = \{j\}$ (and thus $|S| = 1$), we have
      \begin{equation*}
            \max_{j \in S} \left\lvert \mathbf{Z}_{\{j\}} - \mathbb{E}[\tilde{\mathbf{Z}}_{\{j\}}] \right\rvert \leq C\sqrt{\frac{\log n}{m}}.
      \end{equation*}
      Taking $m \geq C\eta^{-2}p^2s^2(p)\log n$ yields the desired result.

\end{proof}

\subsection{Concentration analysis for $\f$ in~\eqref{eq:def-Zf}}
\begin{lemma}\label{conc-f}
      For any constant $\eta > 0$ and $\gamma > 0$, the following concentration inequality
      \begin{equation}
            \max_{l \in [n]} \left\lvert f_{l} - \mathbb{E}(\tilde{f}_{l})\right\rvert < \eta \times \frac{\gamma}{2\sqrt{ks(p)}}
      \end{equation}
      holds with probability exceeding $1 - n^{-c_\eta}$ given $m \geq C_{\eta} {\gamma^{-2}} ks(p)\log n$, where $c_{\eta}$ and $C_{\eta}$ are numerical constant depending on $\eta$.
\end{lemma}

\begin{proof}
      The goal is to bound the element-wise error $\max_{l} |f_l - \mathbb{E}[\tilde{f}_l]|$. We rely on the independence between the measurement vectors $\mathbf{a}_i$ in the current stage and the vector $\mathbf{e}^0$. We decompose the error into two terms as follows:
      \begin{align}
            \nonumber |f_l - \mathbb{E}[\tilde{f}_l]| & \le \underbrace{|f_l - \tilde{f}_l|}_{\text{bias term}} + \underbrace{|\tilde{f}_l - \mathbb{E}[\tilde{f}_l]|}_{\text{fluctuation term}}        \\
                                                      & = |\e_l^* (\mathbf{Z} - \tilde{\mathbf{Z}}) \e^0| + |\e_l^* (\tilde{\mathbf{Z}} - \mathbb{E}[\tilde{\mathbf{Z}}]) \e^0|. \label{eq:f-decompose}
      \end{align}

      On one hand, the bias term is controlled by the operator norm of the difference matrix $\Delta \mathbf{Z} = \mathbf{Z} - \tilde{\mathbf{Z}}$. Note that $\e^0$ is supported on $S^0$, and denote $U = S^0\cup \{l\}$ of size $|U| = p+1 \leq 2k$. Then we have
      \begin{align}
            \nonumber |\e_l^* (\mathbf{Z} - \tilde{\mathbf{Z}}) \e^0|
                      & \eqmakebox[op]{$\le$} \| (\mathbf{Z} - \tilde{\mathbf{Z}})_{U} \| \|\e_l\|_2 \|\e^0\|_2 \\
            \nonumber & \eqmakebox[op]{$\le$} \| (\mathbf{Z} - \tilde{\mathbf{Z}})_{U} \|                       \\
                      & \eqmakebox[op]{$\overset{(a)}{\le}$} C \sqrt{\frac{\log n}{m}},\label{eq:f-decompose1}
      \end{align}
      where (a) follows from \Cref{lemma:op-diff2}.

      On the other hand, conditioned on $S^0$ (and thus $\mathbf{e}^0$), we analyze the fluctuation term $|\tilde{f}_l - \mathbb{E}[\tilde{f}_l]|$. Recall that $\tilde{\mathbf{f}} = \tilde{\mathbf{Z}}\mathbf{e}^0$. We can expand the $l$-th entry as a sum of $m$ independent random variables:
      \begin{equation}
            \tilde{f}_l = \mathbf{e}_l^* \tilde{\mathbf{Z}} \mathbf{e}^0 = \frac{1}{m} \sum_{i=1}^m \kappa_i (\mathbf{a}_i^* \mathbf{e}_l)^* (\mathbf{a}_i^* \mathbf{e}^0),
      \end{equation}
      where $\kappa_i = \frac{1}{2} - \exp(-y_i^2/\|\x\|^2)$. Since we assume independence between the current samples $\{\mathbf{a}_i\}$ and the initialization $\mathbf{e}^0$ (data splitting), the terms in the summation are independent conditioned on $\mathbf{e}^0$. Let $u_i = \mathbf{a}_i^* \mathbf{e}_l$ and $v_i = \mathbf{a}_i^* \mathbf{e}^0$. Both $u_i$ and $v_i$ are standard complex Gaussian variables (potentially correlated if $l \in \text{supp}(\mathbf{e}^0)$). The $i$-th summand is $X_i = \kappa_i \bar{u}_i v_i$. Note that the weights are bounded as $|\kappa_i| \le 1/2$. The product of two Gaussian variables $\bar{u}_i v_i$ follows a sub-exponential distribution with bounded Orlicz norm $\| \bar{u}_i v_i \|_{\psi_1} \le C \|u_i\|_{\psi_2} \|v_i\|_{\psi_2} \le C'$. Consequently, the weighted variable $X_i$ is sub-exponential with
      \begin{equation}
            \| X_i \|_{\psi_1} \le |\kappa_i| \| \bar{u}_i v_i \|_{\psi_1} \le C''.
      \end{equation}
      Applying the scalar Bernstein inequality~\cite[Theorem 2.8.1]{hdpBook} directly, we have
      \[
            \mathbb{P}(|\tilde{f}_l - \mathbb{E}[\tilde{f}_l]| \le t ) \ge 1-2 \exp\left(-cm \min\{t,t^2\}\right).
      \]
      Setting $t = \sqrt{\frac{\log n}{m}}$ and taking the union bound over all $l \in [n]$, we get
      \begin{equation}\label{eq:f-decompose2}
            \max_l |\tilde{f}_l - \mathbb{E}[\tilde{f}_l]| \le C' \sqrt{\frac{\log n}{m}}.
      \end{equation}
      Combining~\eqref{eq:f-decompose1} and~\eqref{eq:f-decompose2}, we have
      \begin{equation}
            |f_l - \mathbb{E}[\tilde{f}_l]| \le C \sqrt{\frac{\log n}{m}}.
      \end{equation}
      Therefore, when $m\ge C\eta^{-2}\gamma^{-2}ks(p) \log n$, the desired result holds with probability exceeding $1-n^{-c_\eta}$.
\end{proof}

\subsection{Proof for Example~\ref{ex:intermediate-spike}}\label{sec:proof-example-intermediate-spike}

For sufficiently large $k$, we have $k^{-\delta}\gg k^{-1/2}$ (since $\delta<1/2$), so the ordering in~\eqref{eq:intermediate-spike-signal} is consistent with $x_{(1)}\ge x_{(2)}\ge\cdots\ge x_{(k)}$.

First, we have
\begin{equation}\label{eq:intermediate-spike-norm}
      \|\x\|^2
      = k^{-2\delta} + (m-1)\cdot \frac{1}{4}k^{-2\delta} + (k-m)\cdot k^{-1}
      = \Theta(1),
\end{equation}
because $(m-1)k^{-2\delta}=\Theta\!\big(k^{1/2-2\delta}\big)=o(1)$ for $\delta>1/4$, and $(k-m)k^{-1}=1-\Theta(k^{-1/2})$.

We now compare the orders of magnitude for the quantities of interest:
\begin{itemize}
      \item For $p=1$,
            \begin{equation}\label{eq:intermediate-spike-s1}
                  s(1)=\frac{\|\x\|^2}{|x_{(1)}|^2}
                  =\Theta(1)\cdot k^{2\delta}
                  =\Theta\!\big(k^{2\delta}\big).
            \end{equation}
      \item For $m=\lceil \sqrt{k}\rceil$,
            \begin{equation}\label{eq:intermediate-spike-sm}
                  \sum_{j=1}^m |x_{(j)}|^2
                  = k^{-2\delta} + (m-1)\cdot \frac{1}{4}k^{-2\delta}
                  = \Theta\!\big(k^{1/2-2\delta}\big).
            \end{equation}
            Thus we have
            \begin{equation}
                  s^2(m)
                  = \Theta\!\big(k^{4\delta-1}\big).
            \end{equation}
\end{itemize}
Consequently, for all $\delta\in(1/4,1/2)$, we have
\begin{equation}\label{eq:intermediate-spike-condition}
      \frac{s^2(m)}{s(1)}
      = \Theta\!\big(k^{2\delta-1}\big)\to 0,
\end{equation}
and hence $\Theta(s^2(\lceil \sqrt{k}\rceil))<\Theta(s(1))$.

\section{Proof for Proposition~\ref{prop:1}}
Denote $\overline{S}$ as the index set corresponding to the largest $p$ elements of $\x$. We establish the upper and lower bounds for $\frac{1}{p} \sum_{j\in S^0} Z_j$ and $\frac{1}{p} \sum_{j\in \overline{S}} Z_j$, respectively. Applying \Cref{conc-uv} yields

\begin{align}
      \frac{1}{p} \sum_{j\in S^0} Z_j \leq & ~ \frac{1}{p}  \sum_{j \in S^0}\mathbb{E}\left[\tilde{Z}_{j}\right] +  \left \lvert \frac{1}{p}  \sum_{j \in S^0} \left( Z_{j} - \mathbb{E}\left[\tilde{Z}_{j}\right] \right) \right \rvert \nonumber \\
      \leq                                 & ~  \frac{\|\x_{S^0}\|^2}{4p \|\x\|^2} + \frac{1}{16ps(p)} \label{lower-v}
\end{align}
and
\begin{align}
      \frac{1}{p} \sum_{j\in \overline{S}} Z_j \geq & ~ \frac{1}{p}  \sum_{j \in \overline{S}}\mathbb{E}\left[\tilde{Z}_{j}\right] -  \left \lvert \frac{1}{p}  \sum_{j \in \overline{S}} \left( Z_{j} - \mathbb{E}\left[\tilde{Z}_{j}\right] \right) \right \rvert \nonumber \\
      =                                             & \frac{\|\x_{\overline{S}}\|^2}{4 p\|\x\|^2} - \left \lvert \frac{1}{p}  \sum_{j \in \overline{S}} \left( Z_{j} - \mathbb{E}\left[\tilde{Z}_{j}\right] \right) \right \rvert \nonumber                                   \\
      \geq                                          & ~ \frac{1}{4ps(p)} - \frac{1}{16ps(p)} = \frac{3}{16ps(p)}.\label{upper-u}
\end{align}
From the definition of $S^0$ and $\overline{S}$, we have
\begin{equation}\label{eq:uandv}
      \frac{1}{p} \sum_{j\in S^0} Z_j \geq  \frac{1}{p} \sum_{j\in \overline{S}} Z_j.
\end{equation}
Therefore, combining~\eqref{lower-v},~\eqref{upper-u} and~\eqref{eq:uandv}, we have
\begin{equation}
      \frac{\|\x_{S^0}\|^2}{4p \|\x\|^2} + \frac{1}{16ps(p)} \geq \frac{3}{16ps(p)},
\end{equation}
which implies
\begin{equation}
      \frac{\|\x_{S^0}\|^2}{\|\x\|^2} \geq \frac{1}{2 s(p)}.
\end{equation}
This completes the proof.

\section{Proof for Proposition~\ref{prop:2}}\label{pf:prop2}
For any set $\TT$ satisfying $|\TT| = p' \leq k$ and $\frac{\|\x_{\TT}\|^2}{\|\x\|^2} \geq \frac{1}{2s(p)}$, let $\z$ be the unit eigenvector corresponding to the largest eigenvalue, denoted as $\tau$, of
\begin{equation*}
      \mathbf{Z}_{\TT} = \frac{1}{m} \sum_{i=1}^{m} \left(\frac{1}{2} - \exp\left(-\frac{y_i^2}{\lambda^2} \right)\right) (\mathbf{a}_i)_{\TT} (\mathbf{a}_i)^{*}_{\TT}.
\end{equation*}
Then, we have
\begin{align}
      \left\lvert \tau \|\z\|^2 - \frac{|\x^{*}_{\TT}\z|^2}{4 \|\x\|^2}\right\rvert \eqmakebox[op]{$=$} & \left\lvert \z^{*}\mathbf{Z}_{\TT}\z -\z^{*}\left( \frac{\x_{\TT}\x^{*}_{\TT}}{4 \|\x\|^2} \right)\z\right\rvert \nonumber                           \\
      \eqmakebox[op]{$\overset{(a)}{\leq}$}                                                             & \left\lVert \mathbf{Z}_{\TT} - \frac{\x_{\TT}\x^{*}_{\TT}}{4 \|\x\|^2} \right\rVert \|\z\|^2 \overset{(b)}{\leq} \frac{\eta}{s(p)} , \label{inner-1}
\end{align}
where $(a)$ is based on the relationship between matrix norm and vector norm,

and $(b)$ employs~\Cref{cor:op-diff}.
The above inequality holds with probability at least $1-n^{-c_\eta}$ when $m \geq C_{\eta} p' s^2(p)\log n$ since $|\TT| = p'$. Therefore, through simple transformation,~\eqref{inner-1} can be expressed as
\begin{equation}
      |\x^{*}_{\TT}\z|^2 \geq 4 (\tau - \frac{\eta}{s(p)}) \|\x\|^2.\label{inner-2}
\end{equation}
Moreover, since $\tau$ is the largest eigenvalue of the Hermitian matrix $\mathbf{Z}_{\TT}$, we can estimate it as
\begin{align}
      \tau ~\eqmakebox[op]{$\geq$}          & ~ \frac{1}{\|\x_{\TT}\|^2} \x^{*}_{\TT} \mathbf{Z}_{\TT} \x_{\TT} \nonumber                                                                                              \\
      \eqmakebox[op]{$=$}                   & ~ \frac{1}{\|\x_{\TT}\|^2} \x^{*}_{\TT} \left(\mathbf{Z}_{\TT} - \frac{\x_{\TT}\x^{*}_{\TT}}{4 \|\x\|^2}  \right) \x_{\TT} + \frac{\|\x_{\TT}\|^2}{4 \|\x\|^2} \nonumber \\
      \eqmakebox[op]{$\overset{(a)}{\geq}$} & ~  -\frac{\eta}{s(p)} + \frac{\|\x_{\TT}\|^2}{4 \|\x\|^2}, \label{inner-3}
\end{align}
where $(a)$ comes from the definition of matrix norm and~\Cref{cor:op-diff}. Taking~\eqref{inner-3} into~\eqref{inner-2} yields
\begin{align}
      |\x^{*}_{\TT}\z|^2 \geq & ~ 4\left( \frac{\|\x_{\TT}\|^2}{4 \|\x\|^2} - 2 \frac{\eta}{s(p)} \right)  \|\x\|^2 \nonumber                        \\
      =                       & ~ \|\x_{\TT}\|^2  - 8 \frac{\eta}{s(p)}  \|\x\|^2  \overset{(a)}{\geq} (1 - 16 \eta) \|\x_{\TT}\|^2, \label{inner-4}
\end{align}
where $(a)$ comes from the assumption $\frac{\|\x_{\TT}\|^2}{\|\x\|^2}  \geq \frac{1}{2s(p)}$. Let $ \sqrt{1 - 16 \eta}=1/2$, we complete the proof.

\section{Proof for Proposition~\ref{prop:3}}
From the definition of $\tilde{f}_l$, we can derive that
\begin{equation}
      |\mathbb{E}[\tilde{f}_l]| = \left\{
      \begin{array}{cc}
            0,                                   & l \notin \xsupp, \\
            \frac{|\x^{*}\z_0||x_l|}{4\|\x\|^2}, & l \in \xsupp.
      \end{array}
      \right. \label{Ef}
\end{equation}
Based on the definition of $S_{\gamma}$, we evaluate the gap between $f_l, l \in S_{\gamma}$ and $f_l, l \notin \xsupp$. For $l \in S_{\gamma}$, we have
\begin{align}
      |f_{l}|~ \eqmakebox[op]{$\geq$}       & ~ |\mathbb{E}[\tilde{f}_l]| - |f_l - \mathbb{E}[\tilde{f}_l]| \nonumber                                                                                                      \\
      \eqmakebox[op]{$\overset{(a)}{\geq}$} & ~ \frac{1}{4\|\x\|^2} \times \frac{1}{2}\sqrt{\frac{1}{2s(p)}} \|\x\| \times \frac{\gamma \|\x\|}{2 \sqrt{k}} - |f_l - \mathbb{E}[\tilde{f}_l]| \nonumber                    \\
      \eqmakebox[op]{$\overset{(b)}{>}$}    & ~ \frac{\gamma }{16 \sqrt{2}} \sqrt{\frac{1}{ks(p)}}  - \frac{\gamma}{32 \sqrt{2}} \sqrt{\frac{1}{ks(p)}} = \frac{\gamma}{32 \sqrt{2}} \sqrt{\frac{1}{ks(p)}}, \label{sga-1}
\end{align}
where $(a)$ is from the definition of $S_{\gamma}$ and Proposition~\ref{prop:2}, and $(b)$ employs Lemma~\ref{conc-f} with $\eta = \frac{1}{16 \sqrt{2}}$ and holds with probability exceeding $1 - n^{-c}$ when $m \geq C \gamma^{-2}ks(p) \log n$.

On the other hand, for $l \notin \xsupp$, we have
\begin{equation}
      |f_{l}| \leq  |\mathbb{E}[\tilde{f}_l]| + |f_l + \mathbb{E}[\tilde{f}_l]| \overset{(a)}{\leq} \frac{\gamma }{32 \sqrt{2}} \sqrt{\frac{1}{ks(p)}}, \label{sga-2}
\end{equation}
where $(a)$ comes from~\eqref{Ef} and Lemma~\ref{conc-f} with $\eta = \frac{1}{16 \sqrt{2}}$. By combining~\eqref{sga-1} and~\eqref{sga-2}, we claim that all indices in $S_{\gamma}$ are selected in this case. Therefore, we can derive that
\begin{equation}
      \frac{\|\x_{S^{1}}\|^2}{\|\x\|^2} \geq 1 - \sum_{i \notin S_{\gamma}} \frac{|x_i|^2}{\|\x\|^2} \geq 1 - \frac{\gamma^2}{4}.
\end{equation}

\section{Proof of Proposition~\ref{prop:4}}

Recall that in \textbf{Step 4} of the algorithm, $\z$ is defined as the eigenvector of $\mathbf{Z}_{S^1}$ corresponding to the largest eigenvalue, scaled such that $\|\z\|_2 = \lambda$, where $\lambda = \sqrt{\frac{1}{m}\sum_{i=1}^m y_i^2}$. Let $\hat{\u} = \z / \|\z\|_2$ be the estimated direction, and $\u^* = \x_{S^1} / \|\x_{S^1}\|_2$ be the true direction on the subset $S^1$.

By the definition of the distance, we have
\begin{align}
      \nonumber     \text{dist}(\z, \x_{S^1})
       & = \min_{\phi \in [0, 2\pi)} \| \z - e^{j\phi} \x_{S^1} \|_2 \nonumber                              \\
       & \le \| \z - e^{j\phi^*} \x_{S^1} \|_2 \nonumber                                                    \\
       & = \Big\| \lambda \hat{\u} - e^{j\phi^*} \|\x_{S^1}\|_2 \u^* \Big\|_2, \label{eq:dist-decomp-start}
\end{align}
where $\phi^*$ is the optimal phase aligning $\hat{\u}$ and $\u^*$. By adding and subtracting $\lambda e^{j\phi^*} \u^*$, we can decompose the error into a norm component and a directional component via the triangle inequality:
\begin{align}
      \nonumber & ~~~~\text{dist}(\z, \x_{S^1})                                                                                                                                                 \\
                & \le \Big\| \lambda \hat{\u} - \lambda e^{j\phi^*} \u^* \Big\|_2 + \Big\| \lambda e^{j\phi^*} \u^* - \|\x_{S^1}\|_2 e^{j\phi^*} \u^* \Big\|_2 \nonumber                        \\
                & = \lambda \underbrace{\| \hat{\u} - e^{j\phi^*} \u^* \|_2}_{\text{directional error}} + \underbrace{| \lambda - \|\x_{S^1}\|_2 |}_{\text{norm error}}. \label{eq:dist-decomp}
\end{align}

On one hand, for the directional error, the expected matrix is $\mathbb{E}[\tilde{\mathbf{Z}}_{S^1}] = \frac{\x_{S^1}\x_{S^1}^*}{4\|\x\|^2}$. The eigengap between the largest eigenvalue and the second largest (which is 0) is $\delta_{gap} = \frac{\|\x_{S^1}\|^2}{4\|\x\|^2}$.
According to the Davis-Kahan $\sin \Theta$ theorem~\cite{davis1970rotation}, the distance between the eigenvectors is bounded by the ratio of the perturbation to the eigengap:
\begin{align}
      \nonumber   \| \hat{\u} - e^{j\phi^*} \u^* \|_2 & \le \frac{\sqrt{2} \| \mathbf{Z}_{S^1} - \mathbb{E}[\tilde{\mathbf{Z}}_{S^1}] \|}{\|\x_{S^1}\|^2 / (4\|\x\|^2)}      \\
                                                      & \overset{(a)}{\le} \frac{4\sqrt{2} \|\x\|^2}{\|\x_{S^1}\|^2} C \sqrt{\frac{k \log n}{m}}, \label{eq:dir-bound-final}
\end{align}
where $(a)$ comes from Corollary~\ref{cor:op-diff}, with probability at least $1-2e^{-ck\log n}$.

On the other hand, for the norm error, we need to bound $|\lambda - \|\x_{S^1}\|_2|$. By the triangle inequality, we have
\begin{equation}\label{eq:norm-error-decompose}
      |\lambda - \|\x_{S^1}\|_2| \le |\lambda - \|\x\|_2| + |\|\x\|_2 - \|\x_{S^1}\|_2|.
\end{equation}
First, from Lemma~\ref{lemma:x-lambda-diff}, we know that $|\lambda^2 - \|\x\|^2| \le C \sqrt{\frac{\log n}{m}} \|\x\|^2$. Using the inequality $|\sqrt{a} - \sqrt{b}| \le \frac{|a-b|}{\sqrt{a}}$, we have
\begin{equation}\label{eq:norm-error-decompose1}
      |\lambda - \|\x\|_2| \le \frac{|\lambda^2 - \|\x\|^2|}{\|\x\|_2} \le C \sqrt{\frac{\log n}{m}} \|\x\|_2.
\end{equation}
When $m > c\log n$, it holds that
\begin{equation}\label{eq:lambda-upper-bound}
      \lambda \le \left(1+C\sqrt{\frac{\log n}{m}}\right) \|\x\|_2 \le 2\|\x\|_2.
\end{equation}
Second, from Proposition~\ref{prop:3}, we know $S^1$ contains the dominant entries $S_\gamma$, implying
\begin{equation}\label{eq:x-xs-ratio}
      \|\x_{S^1}\|_2 \ge \sqrt{1 - \gamma^2/4} \|\x\|_2.
\end{equation}
We can always choose small $\gamma\le 1/2$ such that $\sqrt{1 - \gamma^2/4} \ge 1-\gamma/2$. Thus, we have
\begin{equation}\label{eq:norm-error-decompose2}
      |\|\x\|_2 - \|\x_{S^1}\|_2| \le \frac{\gamma}{2} \|\x\|_2.
\end{equation}

Collecting the above bounds \eqref{eq:dir-bound-final} --~\eqref{eq:norm-error-decompose2} with constants absorbed into $C_1$, $C_2$, we have
\begin{align}
      \nonumber  ~~~~ \text{dist}(\z, \x_{S^1}) & \le C_1 \left(1-\frac{\gamma^2}{4}\right)^{-1} \|\x\|_2 \sqrt{\frac{k \log n}{m}}                    \\
      +~                                        & \left(C_2\sqrt{\frac{\log n}{m}} + \frac{\gamma}{2}\right)\|\x\|_2. \label{eq:norm-error-decompose3}
\end{align}
Choose a small constant $\gamma$ such that $\left(1-\frac{\gamma^2}{4}\right)^{-1} \le 2$. We have
\begin{equation}
      \text{dist}(\z, \x_{S^1}) \le C\sqrt{\frac{k \log n}{m}}\|\x\|_2 + \frac{\gamma}{2}\|\x\|_2.
\end{equation}

Finally, considering the total distance to the original signal $\x$:
\begin{align}
      \text{dist}(\z, \x) & \le \text{dist}(\z, \x_{S^1}) + \|\x_{S^1} - \x\|_2 \nonumber               \\
                          & \le C \|\x\| \sqrt{\frac{k \log n}{m}} + \gamma \|\x\|.\label{eq:dist-comp}
\end{align}
To ensure $\text{dist}(\z, \x) \le \delta \|\x\|$, we set $\gamma < \frac{1}{2} \delta$. \Cref{prop:3} implies a scaling dependency of $m \propto \gamma^{-2}$. Consequently, setting $\gamma < \frac{1}{2} \delta$ requires $m \propto \delta^{-2}$. Crucially, this sample complexity is sufficient to bound the first term in \eqref{eq:dist-comp} by $\frac{1}{2}\delta \|\x\|$ as well. Thus, the condition is satisfied provided that
\begin{equation}
      m \ge C \delta^{-2} k \log n.
\end{equation}

\bibliographystyle{IEEEtran}
\bibliography{IEEEabrv,refs}

\end{document}